\documentclass[letter,11pt]{article}
\usepackage[utf8]{inputenc}
\usepackage{amsfonts}
\usepackage{amssymb}
\usepackage{amsthm}
\usepackage{amsmath}
\usepackage[letterpaper,left=1in,right=1in,top=1in,bottom=1.2in,footskip=0.6in]{geometry}
\usepackage[all]{xy}
\usepackage{marginnote}
\usepackage{bbm}
\usepackage{aliascnt}
\usepackage{hyperref}
\usepackage{cleveref}
\hypersetup{colorlinks=true,citecolor=blue,linkcolor=red,urlcolor=blue}
\usepackage{eurosym}
\usepackage{graphicx}
\usepackage{colortbl}
\usepackage{epstopdf}
\usepackage{array}
\usepackage{xcolor}
\usepackage[nottoc]{tocbibind}  
\usepackage{todonotes}
\usepackage{tikz}
\usepackage{arydshln}
\usepackage{titlesec}
\usepackage{stmaryrd}
\usepackage{comment}

\titleformat{\paragraph}
{\normalfont\normalsize\bfseries}{\theparagraph}{1em}{}
\titlespacing*{\paragraph}
{0pt}{3.25ex plus 1ex minus .2ex}{1.5ex plus .2ex}
\DeclareMathAlphabet{\mathbbold}{U}{bbold}{m}{n}

\newcommand*{\R}{\mathbb{R}}

\newcommand*{\N}{\mathbb{N}}
\newcommand*{\C}{\mathbb{C}}
\newcommand*{\D}{\mathbb{D}}
\newcommand*{\G}{\mathcal{G}}

\newcommand*{\U}{\mathbb{U}}
\newcommand{\ra}\rightarrow
\newcommand{\vbeta}{{\boldsymbol\beta}}

\newcommand{\vlambda}{{\boldsymbol\lambda}}

\newcommand{\vw}{{\boldsymbol w}}
\newcommand{\abs}[1]{\left\lvert #1 \right\rvert}
\newcommand{\bs}{\backslash}
\newcommand{\eps}{\varepsilon}
\newcommand{\vp}{{\boldsymbol{p}}}

\newcommand{\vone}{{\boldsymbol{1}}}

\newcommand{\Ising}{\mathrm{Ising}}
\newcommand{\RC}{\mathrm{RC}}

\newtheorem{theorem}{Theorem}
\newcommand{\NewTheoremWithAlias}[2]{%
  \newaliascnt{#1}{theorem}%
  \newtheorem{#1}[#1]{#2}%
  \aliascntresetthe{#1}%
  \crefname{#1}{\MakeLowercase{#2}}{\MakeLowercase{#2}s}%
  \Crefname{#1}{#2}{#2s}%
}
\NewTheoremWithAlias{lemma}{Lemma}
\NewTheoremWithAlias{corollary}{Corollary}
\NewTheoremWithAlias{proposition}{Proposition}
\NewTheoremWithAlias{definition}{Definition}
\NewTheoremWithAlias{example}{Example}
\NewTheoremWithAlias{remark}{Remark}
\NewTheoremWithAlias{fact}{Fact}
\NewTheoremWithAlias{claim}{Claim}
\NewTheoremWithAlias{conjecture}{Conjecture}
\NewTheoremWithAlias{question}{Question}

\title{\bf Zero-Freeness is All You Need: A Weitz-Type FPTAS for the Entire Lee--Yang Zero-Free Region}
\author{}
\date{}

\author{Shuai Shao\thanks{School of Computer Science and Technology \& Hefei National Laboratory, University of Science and Technology of China.}\\
{\tt shao10@ustc.edu.cn}
\and
Ke Shi\footnotemark[1]\\
{\tt self.ke.shi@gmail.com}}

\begin{document}
\begin{titlepage}
    \maketitle
    \thispagestyle{empty}
\begin{abstract}
 We present a Weitz-type FPTAS for the ferromagnetic Ising model across the entire Lee–Yang zero-free region, without relying on the strong spatial mixing (SSM) property.
Our algorithm is Weitz-type for two reasons.
First, it expresses the partition function as a telescoping product of ratios, with the key being to approximate each ratio.
Second, it uses Weitz’s self-avoiding walk tree, and truncates it at logarithmic depth to give a good and efficient approximation. The key difference from the standard Weitz algorithm is that we approximate a carefully designed edge-deletion ratio instead of the marginal probability of a vertex being assigned a particular spin, ensuring our algorithm does not require SSM.

Furthermore, by establishing local dependence of coefficients (LDC), we indeed prove a novel form of SSM for these edge-deletion ratios, which, in turn, implies the standard SSM for the random cluster model. This is the first SSM result for the random cluster model on general graphs, beyond lattices. Our proof of LDC is based on a new divisibility relation, and we show such relations hold quite universally. 
This leads to a broadly applicable framework for proving LDC across a variety of models, including the Potts model, the hypergraph independence polynomial, and Holant problems. Combined with existing zero-freeness results for these models, we derive new SSM results for them.
\end{abstract}
\end{titlepage}
\newpage

\setcounter{tocdepth}{2}
\newpage

\section{Introduction}
A \emph{fully polynomial-time approximation scheme} (FPTAS) is a family of algorithms $\{A_\varepsilon\}$, where  $A_\varepsilon$ is a multiplicative $(1\pm\varepsilon)$-approximation algorithm with running time polynomial in
both the input size and $1/\varepsilon$ for each $\varepsilon>0$.
For counting problems, a standard approach to designing FPTASes is based on complex zero-free regions of the associated \emph{partition function}. Once such a zero-free region is established, Barvinok's algorithm~\cite{barvinok2016combinatorics} provides an FPTAS for approximating the partition function in a slightly smaller region.
Specifically, suppose the partition function $Z$ has no zeros in a complex region that contains a computationally tractable point. 
Then, possibly after a complex conformal mapping, the zero-freeness property ensures that $\log Z$ can be well-approximated in a slightly smaller region by its Taylor expansion series truncated at degree $O(\log n)$ where $n$ is the instance size. More precisely, the Taylor expansion series $f_k$ of $\log Z$ truncated at degree $k$ gives an approximation of $\log Z$ within additive error $Cr^{-k}$ for some positive constant $C$ and $r>1$. 
The coefficients of $f_k$ can be computed via subgraph counting in time $\Delta^{O(k)}$~\cite{patel2017deterministic} where $\Delta$ is the maximum degree. Clearly, the running time is polynomial on $n$ when $k=O(\log n)$.
This approach connects the long-standing study of complex zeros of the partition function in statistical physics to algorithmic studies. 

Another (and earlier) approach for devising FPTASes, originating in the work of Weitz~\cite{Weitz06} and independently in Bandyopadhyay and Gamarnik~\cite{bandyopadhyay2008counting}, relies on the \emph{correlation decay} property, or more precisely, \emph{strong spatial mixing} (SSM). Roughly speaking, SSM asserts that correlations between spins decay exponentially with distance. Weitz’s algorithm approximates the partition function defined on a graph $G$ by expressing it as a telescoping product of certain marginal probabilities. The key technical ingredient of Weitz's algorithm is the \emph{self-avoiding walk} (SAW) tree, which reduces the computation of a marginal probability on the original graph $G$ to that on the SAW tree. 
However, the SAW tree may be exponentially large compared to the graph size $n$. 
The SSM property guarantees that the marginal probability  can be well approximated by truncating the SAW tree at a depth of $O(\log n)$, making the evaluation efficient. 

Both Barvinok’s algorithm and Weitz’s algorithm have been widely applied, especially to the study of 2-spin systems, which are among the most fundamental and well-studied models in statistical physics and counting problems.
A 2-spin system is defined on a finite simple graph $G=(V,E)$ with parameters $(\beta,\gamma,\lambda)$: two edge activities $\beta,\gamma$ representing edge interactions, and a vertex activity $\lambda$ representing an external field. 
A partial configuration of this system refers to a mapping $\sigma: \Lambda\to\{+,-\}$ for some $\Lambda\subseteq V$ which may be empty.
When $\Lambda=V$, it is a configuration and is assigned a weight
$w(\sigma)=\beta^{m_+(\sigma)}\gamma^{m_-(\sigma)}\lambda^{n_+(\sigma)},$
where $m_+(\sigma)$ and $m_-(\sigma)$ count $(+,+)$ and $(-,-)$ edges respectively, and $n_+(\sigma)$ counts vertices with spin $+$. 
The associated partition function is
$Z_G(\beta,\gamma,\lambda) = \sum_{\sigma:V\to\{+,-\}} w(\sigma).$
Many natural combinatorial problems reduce to evaluating $Z_G(\beta,\gamma,\lambda)$. For instance, the case $(\beta=0,\gamma=1)$ corresponds to the hard-core model (independence polynomial), while $\beta=\gamma$ gives the celebrated Ising model. Depending on whether $\beta\gamma>1$ or $\beta\gamma<1$, the model is classified as ferromagnetic or antiferromagnetic, respectively.

Although FPTASes for 2-spin systems have been obtained via both Barvinok's algorithm \cite{peters2019conjecture,Bencszero,mann2019approximation,liu2019ising,LSSFisherzeros,petersregts20,GGHP22,PRS23,zerofe, shao2021contraction} and Weitz's algorithm \cite{zhang2009approximating,li2012approximate,li2013correlation,sinclair2014approximation},  the applicability differs. While 
Barvinok’s algorithm covers broad regions including ferromagnetic systems, Weitz’s algorithm is mainly effective for antiferromagnetic systems where SSM holds.
The SSM property crucially required by Weitz's algorithm is often absent in the ferromagnetic regime. 
Recent work~\cite{regts2023absence, shao2024zero} established a connection between these two frameworks by showing that zero-freeness implies SSM, provided zero-free results hold for graphs with pinned vertices. 
As a consequence, some new SSM results have been proved for 2-spin systems~\cite{shao2024zero}, which makes Weitz’s algorithm can be applied. 
However, some of the most celebrated zero-freeness results, such as the Lee–Yang theorem~\cite{LeeYang1,LeeYang2} for the ferromagnetic Ising model, only hold for graphs without pinned vertices.
 Consequently, for the ferromagnetic Ising model on graphs of bounded degree, 
 although Barvinok’s algorithm yields an FPTAS  throughout the Lee–Yang zero-free region~\cite{liu2019ising}, i.e., $\lambda\in \mathbb{C}$ and $|\lambda|< 1$ or $|\lambda|> 1$ symmetrically,
Weitz's algorithm cannot be applied to the entire zero-free region due to the lack of SSM. 
So far, the best known SSM results hold for regions much smaller 
than the Lee--Yang zero-free region \cite{shao2024zero,shao2021contraction}, namely
the union of the open disk centered at $0$ with radius $\beta^{-\Delta}$ where $\Delta$ is the degree bound 
and a strip-shaped neighborhood of the real segment
$[0, \frac{1}{\beta^{\Delta-2}\,((\Delta-2)\beta^2-\Delta)}).$
In fact, it is known that SSM does not hold throughout the entire zero-free region
\cite{BAG07,sinclair2014approximation}. For instance, for the three-dimensional Ising model at low temperatures where the Lee--Yang theorem holds, it is known that SSM does not hold~\cite{BAG07}, although weak spatial mixing does. 
Although these deterministic-approximation approaches face inherent limitations,
it is worth noting that the ferromagnetic Ising model also admits classical
polynomial-time approximate sampling algorithms via Markov chain Monte Carlo,
most notably the Jerrum–Sinclair algorithm~\cite{jerrum1993polynomial}.

So far,  for 2-spin systems, the regions where Barvinok’s algorithm applies strictly contain and are much larger than those accessible to Weitz’s algorithm. 
This raises a natural and interesting question: \textit{Is Barvinok’s algorithm inherently more powerful than Weitz’s algorithm?} In this paper, we provide negative evidence for this question.

\subsection*{Our contributions}
\begin{theorem}[Informal]
There is a Weitz-type FPTAS for the ferromagnetic Ising model throughout the entire Lee--Yang zero-free region. 
The algorithm does \emph{not}  require SSM.
\end{theorem}
Our algorithm is a Weitz-type algorithm for two reasons.
First, it expresses the partition function as a telescoping product of certain ratios and the key is to approximate each ratio.
Secondly, in order to give a good approximation of the ratios, it uses the SAW tree and truncates it at  logarithmic depth.
However, crucial differences distinguish our algorithm from the standard Weitz algorithm, ensuring that our algorithm does \emph{not} rely on SSM.
First, instead of approximating the marginal probability $P_v=\frac{Z_{G,\sigma(v)=+}}{Z_G}$ of a vertex $v$ being assigned spin $+$, we approximate a carefully designed edge-deletion ratio $P_e=\frac{Z_{G-e}}{Z_G}$ where $G-e$ denotes the graph obtained from $G$ by removing an edge $e$. 
Second, 
since SSM is unavailable, we cannot argue that truncating the SAW tree yields a good approximation.
Inspired by Barvinok’s algorithm, we show that each edge-deletion ratio viewed as a function on $\lambda$ can be well approximated by truncating its Taylor expansion series at logarithmic degree.

A key difference from Barvinok’s algorithm is that, while in the Barvinok framework the low-order Taylor coefficients
are computed via subgraph counting~\cite{patel2017deterministic,liu2019ising}, we compute these coefficients still via recursion on the SAW tree.
The idea of using SAW-tree recursions to compute low-order coefficients has also been explored by Bencs and Regts~\cite{bencs2025barvinok}, who compute low-order Taylor coefficients of $\log Z$.
We note that their method remains Barvinok-type, as their goal is to approximate $\log Z$. 
In contrast, our approach is a genuinely Weitz-style telescoping algorithm that directly approximates each ratio using the SAW tree.

The replacement of $P_v$ by $P_e$ is crucial for the above method to work. 
In fact, it is impossible to show that $P_v$ could be approximated by logarithmic-degree Taylor truncation, since this would imply SSM throughout the Lee–Yang region contradicting known impossibility results~\cite{BAG07}. 
The obstruction comes from the fact that, as shown in \cite{shao2024zero}, the ratio $P_v$ viewed as a function on $\lambda$ and its SAW-tree version $P_v^{T_k}$ truncated at depth $k$ share the same first $k$ coefficients, known as \emph{local dependence of coefficients} (LDC).
Hence,  if the Taylor expansion series $f_k$ truncated at degree $k$ approximates $P_v$ well, i.e., $|P_v-f_k|\leq Cr^{-k}$ for some positive constants $C$ and $r>1$, then $f_k$ also approximates $P_v^{T_k}$ well. 
Then, by a triangle inequality, one would have $|P_v-P_v^{T_k}|\leq 2Cr^k$, which is the standard SSM for the Ising model. 

The above argument further implies that if LDC can be established for edge-deletion ratios, then together with zero-freeness, which guarantees good approximation by Taylor series of logarithmic degree, we obtain a form of SSM for these ratios. 
In other words, zero-freeness plus LDC implies SSM. 
In this paper, we further show that such a LDC property indeed holds for the Ising model with arbitrary parameters $\beta$ and $\lambda$, including regimes beyond the Lee-Yang zero-free region. 
Thus, together with the Lee-Yang theorem, we establish SSM for edge-deletion ratios\footnote{We actually prove a more general form of SSM; see \Cref{thm:eSSM}.} across the entire zero-free region.
In summary,  for Weitz-type FPTASes, zero-freeness alone suffices, while for edge-deletion SSM, one additionally needs LDC, a combinatorial property independent of zero-freeness.

\begin{theorem}[SSM for edge-deletion ratios]\label{thm:edge-deletion-ssm}
Fix $\beta > 1$ and $\lambda \in \D$. Then there exist constants 
$C > 0$ and $r > 1$ such that for every graph $G=(V,E)$, for any edge 
$e \in E$ and subsets $A,B \subseteq E \setminus \{e\}$, we have
\[
    \left| P_{G-A,e} - P_{G-B,e} \right|
    \;\leq\; C r^{-\,d_G(e,\,A \neq B)},
\]
where $P_{G,e}$ denotes the edge-deletion ratio 
$Z_{G-e} / Z_G$ and $d_G(e, A \neq B)$ is the distance in $G$ 
between the edge $e$ and the set of edges $(A\backslash B)\cup(B\backslash A)$.
\end{theorem}

Quite remarkably, the choice of the edge-deletion ratio is not only crucial for a Weitz-type FPTAS for the Ising model, but it 
also admits a probabilistic interpretation in the Ising related random cluster model. 
It corresponds exactly to the marginal probability that an edge is included in the random cluster model. 
Thus, our edge-deletion SSM implies standard SSM for the random cluster model.  
This is the first SSM result for the random cluster model on general graphs, whereas all previous results were confined to lattices. 
As brief background, the random cluster model associated the ferromagnetic Ising model is parameterized by edge parameters $p$ and vertex parameters $\lambda$.
For an edge $e\in E$ and a partial edge configuration $\sigma_{\Lambda}$ specified on $\Lambda\subseteq E(G)$, let
$P^{\sigma\Lambda}_{G,e}$
denote the marginal probability under the random cluster distribution that edge $e$ is excluded conditioned on $\sigma_\Lambda$.

\begin{theorem}[SSM for the random cluster model]
    Fix the parameters of the random cluster model $p \in [0,1]$ and $\lambda \in [0,1)$. 
   	Then there exist constants $C>0$ and $r>1$ such that for any $G=(V,E)$, any edge $e \in E$, any partial configuration
	$\sigma_{\Lambda_1}$ and $\tau_{\Lambda_2}$ where $\Lambda_1,\Lambda_2 \subseteq E\backslash e$, we have
	\[	
		\abs{P_{G,e}^{\sigma_{\Lambda_1}} - P_{G,e}^{\tau_{\Lambda_2}}} 
		\leq Cr^{-d_G(e,\sigma_{\Lambda_1} \neq \tau_{\Lambda_2})}.
	\]
    where $(\sigma_{\Lambda_1} \neq \tau_{\Lambda_2})= (\Lambda_1\setminus\Lambda_2)\cup (\Lambda_2\setminus\Lambda_1)\cup \{f\in\Lambda_1\cap\Lambda_2:\sigma_{\Lambda_1}(f)\neq\tau_{\Lambda_2}(f)\}$.
\end{theorem}

Beyond introducing the notion of edge-deletion SSM, a main technical contribution of this paper is a broadly applicable method for establishing LDC.
The LDC property was first  shown implicitly in~\cite{regts2023absence} using cluster expansion, an analytic method applicable to the hard-core model,  and  was later formally introduced and extended  to general 2-spin systems~\cite{shao2024zero} using a combinatorial approach based on a Christoffel-Darboux type identity. 
In this paper, we substantially generalize the combinatorial framework by circumventing the need of explicit identities: a more applicable division relation proved via a delicate one-to-one correspondence already suffices for LDC.  
We show that such division relations hold broadly across diverse models including the Potts model, the hypergraph independence polynomial, and binary symmetric Holant problems. 
For each model, combining division relations with existing zero-free regions yields new SSM results,  thereby broadening the applicability of correlation-decay techniques.

As an example, we prove the following LDC and SSM properties for the hypergraph independence polynomial.
We 
define $P_{H,v}^{\sigma_{\Lambda}}(\lambda)$ as the standard marginal probability of the vertex $v$ being occupied in hypergraph $H$ under the partial configuration $\sigma_{\Lambda}$ on $\Lambda\subseteq V(H)$.
\begin{lemma}[LDC for the hypergraph independence polynomial]\label{lem:hyper-LDC-intro}
    Let $H=(V,E)$ be a hypergraph, $\sigma_{\Lambda_1},\tau_{\Lambda_2}$ be two partial configurations on $\Lambda_1,\Lambda_2 \subseteq V$,
    $v$ be a proper vertex to $\sigma_{\Lambda_1}$ and $\tau_{\Lambda_2}$, then 
    the formal power series of $P_{H,v}^{\sigma_{\Lambda_1}}(\lambda)$ and 
    $P_{H,v}^{\tau_{\Lambda_2}}(\lambda)$ agree in their first
    $d_H(v,\sigma_{\Lambda_1} \neq \tau_{\Lambda_2})$ coefficients.
\end{lemma}
Combining this LDC property with the optimal zero-free region of \cite{bencs2025optimal}, 
where $\lambda_c(\Delta) = \frac{(\Delta-1)^{\Delta-1}}{(\Delta-2)^{\Delta}}$ and 
$\lambda_s(\Delta) = \frac{(\Delta-1)^{\Delta-1}}{\Delta^{\Delta}}$,
we obtain the following new SSM result.
\begin{theorem}[SSM for the hypergraph independence polynomial]\label{thm:hyper-ssm-intro}
    Fix $\Delta \geq 3$ and $\lambda \in \D_{\lambda_s(\Delta)} \cup (0,\lambda_c(\Delta))$. Then there
     exist constants $C>0$ and $r>1$ such that for any hypergraph $H=(V,E)$ with maximum degree at most $\Delta$, any 
     two feasible partial configurations $\sigma_{\Lambda_1}$ and $\tau_{\Lambda_2}$ where $\Lambda_1$ may differ from $\Lambda_2$, 
     and any vertex $v$ proper to $\sigma_{\Lambda_1}$ and $\tau_{\Lambda_2}$, we have 
    \[
        | P_{H,v}^{\sigma_{\Lambda_1}}(\lambda) - P_{H,v}^{\tau_{\Lambda_2}}(\lambda) | \leq C r^{-d_H(v,\sigma_{\Lambda_1} \neq \tau_{\Lambda_2})}
    \] 
    where $(\sigma_{\Lambda_1} \neq \tau_{\Lambda_2})= (\Lambda_1\setminus\Lambda_2)\cup (\Lambda_2\setminus\Lambda_1)\cup \{v\in\Lambda_1\cap\Lambda_2:\sigma_{\Lambda_1}(v)\neq\tau_{\Lambda_2}(v)\}$.
\end{theorem}

	
\subsection*{Organization}

The paper is organized as follows.
In \Cref{sec:weitz-alg}, to derive the Weitz-type algorithm, we show that the truncated series provides a good approximation of the edge-deletion ratio using zero-freeness, and then analyze the truncated series through the self-avoiding walk tree and operations on formal power series. This yields the FPTAS.
In \Cref{sec:edge-ssm}, we introduce the framework in which zero-freeness implies SSM via LDC, and establish the SSM property in terms of edge activities for the ferromagnetic Ising model, using the Lee–Yang theorem and the divisibility relation. This edge-based SSM property further implies the standard SSM of the corresponding random-cluster model.
In \Cref{sec:other-applications}, we extend the divisibility relation to various models and derive their SSM properties.

\section{Preliminaries}\label{sec:pre}

\subsection{Ising model}

For a graph $G=(V,E)$ with edge activities $\vbeta=(\beta_e)_{e\in E}$ and vertex activities $\vlambda=(\lambda_v)_{v\in V}$, the weight of a configuration $\sigma:V\to\{+,-\}$ is defined by
\[
    w(\sigma) =  \prod_{e \in m(\sigma)} \beta_e \prod_{v \in n(\sigma)} \lambda_v
\]
where $m(\sigma) = \{ e = (u, v) \in E \mid \sigma_u = \sigma_v \}$ is the set of edges whose endpoints have the same spin, 
and $n(\sigma) = \{ v \in V \mid \sigma_v = + \}$ is the set of vertices assigned the $+$ spin. The partition function of
the Ising model is defined by $Z_G(\vbeta, \vlambda) = \sum_{\sigma : V \to \{+, -\}}  w(\sigma)$.
The celebrated Lee--Yang theorem states the zero-free region of the Ising model.
\begin{theorem}[Lee--Yang theorem]
    Let $G=(V,E)$ be a graph with parameters $\vbeta \in [1,\infty)^E$ 
    and $\vlambda \in \D^{V}$ where $\D$ is the open unit disk in the complex plane.
    Then the partition function of the Ising model $Z_G(\vbeta,\vlambda)\neq 0$.
\end{theorem}

A partial configuration of the Ising model is a mapping $\sigma : \Lambda \to \{+,-\}$ for some $\Lambda \subseteq V$,
which may be empty. The conditional partition function is defined as 
$Z_G^{\sigma_\Lambda}(\vbeta, \vlambda) = \sum_{\substack{\sigma : V\to \{+,-\} \\ \sigma |_ \Lambda = \sigma_\Lambda}}$ where 
$\sigma |_ \Lambda$ denotes the restriction of the configuration $\sigma$ on $\Lambda$. Let $u,v \in V$,
then define 
\[ 
Z_{G,v}^{\sigma_\Lambda,+}(\vbeta, \vlambda) = \sum_{\substack{\sigma : V\to \{+,-\} \\ \sigma |_ \Lambda = \sigma_\Lambda, \sigma(v)=+}} w(\sigma), \quad \text{and} \quad
Z_{G,v,u}^{\sigma_\Lambda,+,+}(\vbeta, \vlambda) = \sum_{\substack{\sigma : V\to \{+,-\} \\ \sigma |_ \Lambda = \sigma_\Lambda, \sigma(v)=\sigma(u)=+}} w(\sigma).
\]
Conditional partition functions for the remaining pinning choices
$Z_{G,v}^{\sigma_\Lambda,-}(\vbeta, \vlambda)$, 
$Z_{G,v,u}^{\sigma_\Lambda,+,-}(\vbeta, \vlambda)$,
$Z_{G,v,u}^{\sigma_\Lambda,-,+}(\vbeta, \vlambda)$
and
$Z_{G,v,u}^{\sigma_\Lambda,-,-}(\vbeta, \vlambda)$
are defined analogously.
Then the conditional marginal probability that $v$ is pinned to $+$ given the partial configuration $\sigma_\Lambda$ and the corresponding marginal ratio are defined as
\[
P_{G,v}^{\sigma_\Lambda}(\vbeta,\vlambda) = 
\frac{Z_{G,v}^{\sigma_\Lambda,+}(\vbeta, \vlambda)}{Z_G^{\sigma_\Lambda}(\vbeta, \vlambda)}, \quad \text{and} \quad
R_{G,v}^{\sigma_\Lambda}(\vbeta,\vlambda) =
\frac{Z_{G,v}^{\sigma_\Lambda,+}(\vbeta, \vlambda)}
{Z_{G,v}^{\sigma_\Lambda,-}(\vbeta, \vlambda)}.
\]

\subsection{Weitz's tree reduction}
In the seminal work of Weitz \cite{Weitz06}, the self-avoiding walk (SAW) tree reduces  
the computation of marginal probabilities for 2-spin models on general graphs to the  
corresponding computation on trees.  
We do not repeat the construction of the SAW tree here, and refer readers to \cite{Weitz06} for details.  
We only state the key property of the SAW tree. If the graph $G=(V,E)$  
has maximum degree $\Delta$, then the SAW tree $T_{\text{SAW}}(G,v)$ rooted at $v\in V$  
also has maximum degree $\Delta$.  
Moreover, the marginal probability and marginal ratio at $v$ in $G$ coincide with 
those at the root of $T_{\text{SAW}}(G,v)$ (with some leaves possibly pinned).

Consider a rooted tree $T=(V,E)$ with root $r\in V$. Suppose $r$ has $d$ children
$v_1,\ldots,v_d$. Removing $r$ yields $d$ subtrees $T_1,\ldots,T_d$, where $T_i$ is rooted at
$v_i$. Let $R_{T_i,v_i}$ denote the marginal ratio at $v_i$ in $T_i$ (e.g., $R_{T_i,v_i}
= Z^{+}_{T_i,v_i}/Z^{-}_{T_i,v_i}$), and let $R_{T,r}$ be the marginal ratio at $r$ in $T$.
Then the tree recursion expressing $R_{T,r}$ in terms of the $R_{T_i,v_i}$ is given by a
multivariate map $F_d:\hat{\C}^d \to \hat{\C}$:
\[
  R_{T,r} \;=\; F_d(R_{T_1,v_1},\ldots,R_{T_d,v_d}),\qquad
  F_d(x_1,\ldots,x_d) \;=\; \lambda \prod_{i=1}^d \frac{\beta x_i + 1}{x_i + \gamma},
\]
where $\hat{\C}=\C\cup\{\infty\}$ is the extended complex plane, $\beta,\gamma$ are the
edge-interaction parameters, and $\lambda$ is the external field. 
For the ferromagnetic Ising model, we have $\beta=\gamma > 1$.
If the root $v_i$ of 
$T_i$ is pinned to $+$ or $-$, then we set $R_{T_i,v_i}=\infty$ or $0$, and 
the term $(\beta R_{T_i,v_i}+1)/(R_{T_i,v_i}+\gamma)$ is interpreted as $\beta$ or $1/\gamma$.

Weitz's algorithm \cite{Weitz06} approximates the partition function of the 2-spin system on a graph $G$ by a telescoping product of marginal probabilities. Then 
the SAW tree reduction reduces the problem to approximating marginal probabilities on trees. The strong spatial mixing property on trees guarantees that the marginal probability at the root can be approximated by truncating the tree at logarithmic depth.
The standard strong spatial mixing of the Ising model is given below.

\begin{definition}[Strong spatial mixing]
	Fix parameters $\beta,\lambda$ and a family of graphs $\G$. The Ising model defined on $\G$
	with parameters $(\beta,\lambda)$ is said to satisfy strong spatial mixing with exponential rate $r>1$ if
	there exists a constant $C$ such that for any $G=(V,E) \in \G$, any vertices $v\in V$, any partial configuration
	$\sigma_{\Lambda_1}$ and $\tau_{\Lambda_2}$, we have
	\[	
		\abs{P_{G,v}^{\sigma_{\Lambda_1}}(\beta,\lambda) - P_{G,v}^{\tau_{\Lambda_2}}(\beta,\lambda)} 
		\leq Cr^{-d_G(v,\sigma_{\Lambda_1} \neq \tau_{\Lambda_2})}
	\]  
	where $\sigma_{\Lambda_1} \neq \tau_{\Lambda_2}$ denotes the set 
	$(\Lambda_1 \backslash \Lambda_2) \cup (\Lambda_2 \backslash \Lambda_1) \cup 
	\{ v \in \Lambda_1 \cap \Lambda_2 : \sigma_{\Lambda_1}(v)\neq \tau_{\Lambda_2}(v) \}$, i.e., the set of vertices
	where $\sigma_{\Lambda_1}$ and $\tau_{\Lambda_2}$ differ. The term $d_G(v,\sigma_{\Lambda_1} \neq \tau_{\Lambda_2})$ denotes the shortest path distance from $v$ to any vertex in $\sigma_{\Lambda_1} \neq \tau_{\Lambda_2}$.
\end{definition}

However, the best known SSM results for the ferromagnetic Ising model with $\beta>1$ 
apply only to graphs of bounded degree $\Delta$, and they hold in a region much smaller 
than the Lee--Yang zero-free region \cite{shao2024zero,shao2021contraction}, namely
the union of the open disk centered at $0$ with radius $\beta^{-\Delta}$ 
and a neighborhood of the real segment
$
    \Bigl[0, \dfrac{1}{\beta^{\Delta-2}\,\bigl((\Delta-2)\beta^2-\Delta\bigr)}\Bigr).
$

Consequently, Weitz's algorithm cannot be applied to obtain an FPTAS for the ferromagnetic Ising model over the entire Lee--Yang zero-free region.
Nonetheless, by analyzing the edge deletion ratios and leveraging zero-freeness via truncated Taylor expansions together with the SAW-tree reduction, we establish a Weitz-type FPTAS that works throughout the full Lee--Yang zero-free region.


\subsection{Complex analysis tools and truncated power series}

A \emph{region} is a connected open set in $\C$. 
In particular, an open disk with one interior point removed is also a region. 
We denote by $\D_\rho(z_0)$ the open disk centered at $z_0$ with radius $\rho$, 
and by $\partial \D_\rho(z_0)$ its boundary circle. 
If $z_0 = 0$, we simply write $\D_\rho$ and $\partial \D_\rho$; 
if $\rho = 1$, we omit the subscript $\rho$.
Let $f(z) = \sum_{k \ge 0} a_k z^k $
be a (formal) power series, and write its truncation to degree $m$ as
\[
  f^{[m]} := \sum_{k=0}^m a_k z^k \quad (= f \bmod z^{m+1}).
\]

The following lemma is a standard result of complex analysis textbook, 
derived from Cauchy's integral formula.

\begin{lemma}\label{lem:truncation-error}
Let $f(z)$ be analytic in a neighborhood of $z=0$, and suppose 
$|f(z)| \leq M$ on the circle $\partial \D_\rho$ for some $\rho>0$. 
Then for any $z \in \D_\rho$, we have
\[
    |f(z) - f^{[n]}(z)| 
    \;\leq\; \frac{M}{\rho(r-1) r^n},
    \qquad \text{where } r=\rho/|z|>1 .
\]
\end{lemma}

Applying \Cref{lem:truncation-error} requires a uniform bound on a circle for a family of analytic functions.  
In \cite{regts2023absence}, Regts employed Montel's theorem to obtain such bounds, leading to the following result.

\begin{lemma}\label{lem:montel-bound}  
Let \( \U \) be a region, and let \( \mathcal{F} \) be a family of holomorphic functions \( f: \U \to \mathbb{C} \) such that \( f(\U) \subseteq \mathbb{C} \setminus \{0,1\} \) for all \( f \in \mathcal{F} \).  If there exist a point \( z_0 \in \U \) and a constant \( C \) such that \( |f(z_0)| \leq C \) for all \( f \in \mathcal{F} \), then for any compact subset \( S \subseteq \U \), there exists another constant \( C' \) such that for all \( f \in \mathcal{F} \) and all \( z \in S \), we have \( |f(z)| \leq C' \).  
\end{lemma}

Next, assume the first $m{+}1$ coefficients of $f$ and $g$ are given. 
We measure running time in terms of arithmetic operations over $\C$.  
Using FFT-based polynomial multiplication (see \cite[Secs.~8.2 and~9.1]{von2003modern}), 
the following bounds hold:

\begin{enumerate}
  \item {Scalar multiplication:} $(k f)^{[m]} = k\,f^{[m]}$ in $O(m)$ time.
  \item {Addition:} $(f+g)^{[m]} = f^{[m]} + g^{[m]}$ in $O(m)$ time.
  \item {Multiplication:} $(fg)^{[m]}$ in $O(m\log m)$ time (via FFT).
  \item {Division:} if $g(0)\neq 0$, then $(f/g)^{[m]}$ in $O(m\log m)$ time by Newton iteration.
\end{enumerate}

\section{Weitz-type algorithm for the ferromagnetic Ising model}\label{sec:weitz-alg}


Our approach is a telescoping algorithm based on edge deletion. For a graph $G=(V,E)$ with parameters $(\beta,\lambda)$, we order the edges in $E$ as $e_1,e_2,\ldots,e_m$, denote $G_i = (V,E_i)$ where $E_i = \{e_1,e_2,\ldots,e_i\}$ for $1\leq i \leq m$ and $G_0 = (V,\varnothing)$.
Then we have
\[
    \frac{1}{Z_G(\lambda)} 
    = \prod_{i=1}^m \frac{Z_{G_{i-1}}(\lambda)}{Z_{G_{i}}(\lambda)} \frac{1}{Z_{G_0}(\lambda)} 
    = (1+\lambda)^{-|V|} \prod_{i=1}^m {P_{G_i,e_i}(\lambda)},
\]
where we define the edge–deletion ratio $P_{G,e} = Z_{G-e}/Z_G$.
Thus, approximating $Z_G(\lambda)$ within a multiplicative error of $\varepsilon$ 
reduces to approximating each ratio $P_{G_i,e_i}(\lambda)$ within a multiplicative error 
of at most $\varepsilon/(4m)$, for all $1 \le i \le m$. 
This is achieved by computing the truncated Taylor series of $P_{G_i,e_i}(\lambda)$ 
at $\lambda = 0$ up to degree $k = O(\log(m/\varepsilon))$, 
and then evaluating it at $\lambda$.
\subsection{Truncated series is a good approximation}

To show the truncated series gives a good approximation via \Cref{lem:truncation-error}, we apply \Cref{lem:montel-bound} to obtain a uniform bound on $P_{G,e}(\lambda)$ for all graphs $G$ and edges $e\in G$. This requires that $P_{G,e}(\lambda)$ avoids $0$ and $1$ for all graphs $G$ and edges $e\in G$.

\begin{lemma}\label{lem:edgeremoveavoid01}
Let $G=(V,E)$ be a graph. For edge activity $\beta>1$ and external field $\lambda\in\D$, 
the edge-deletion ratio $P_{G,e}(\beta,\lambda)$ omits the values $0$ and $1$.
\end{lemma}


\begin{proof}
This is a special case of \Cref{lem:edgeavoid01}, where a more general statement
is proved using the Lee--Yang zero-free region.
\end{proof}

Then a uniform bound on $P_{G,e}(\lambda)$ for all graphs $G$ and edges $e\in G$ follows from \Cref{lem:montel-bound}, where the upper bound (for a circle) is used to establish the additive error in \Cref{lem:truncation-error}, and the lower bound (for a single point) is used to turn the additive error into a multiplicative error.

\begin{lemma}\label{lem:edgeremovebound}
    Fix $\beta > 1$, and let $S \subseteq \D$ be a compact set.
    There exist constants $M,b > 0$ such that $b \leq |P_{G,e}(\lambda)| \leq M$ for all graphs $G$, edge $e\in G$ and $\lambda \in S$.
\end{lemma}

\begin{proof}
Fix $\beta>1$ and a compact $S\subseteq\D$.  
Let $\mathcal{F}=\{P_{G,e}(z): G \text{ a graph},\, e\in E(G)\}$.  
By Lemma~\ref{lem:edgeremoveavoid01}, every $f\in\mathcal{F}$ omits $\{0,1\}$ on $\D$, and
$f(0)=1/\beta$ is uniformly bounded. Hence, by Lemma~\ref{lem:montel-bound}, there exists
$M>0$ such that $|P_{G,e}(z)|\le M$ for all $z\in S$, all $G$, and all $e$.

For the lower bound, apply the same argument to
$\mathcal{F}^{-1}=\{1/P_{G,e}(z)\}$. Each $g\in\mathcal{F}^{-1}$ also omits $\{0,1\}$ and
$g(0)=\beta$ is uniformly bounded, so there exists $M'>0$ with
$|1/P_{G,e}(z)|\le M'$ for all $z\in S$. Setting $b:=1/M'$ yields
$b\le |P_{G,e}(z)|\le M$ for all $z\in S$, as claimed.
\end{proof}

\begin{lemma}\label{lem:truncation-relative}
Fix $\beta>1$ and $\lambda \in \D$. 
Then there exists $k=O(\log(m/\varepsilon))$ such that for every graph $G=(V,E)$ with $m=|E|$ and every edge $e\in E$, the truncated series $P_{G,e}^{[k]}$ evaluated at $\lambda$, satisfies the relative bound
\[
  \frac{\bigl|P_{G,e}(\lambda)-P_{G,e}^{[k]}(\lambda)\bigr|}{\bigl|P_{G,e}(\lambda)\bigr|}
  \;\le\; \frac{\varepsilon}{4m}.
\]
\end{lemma}

\begin{proof}
Pick $\rho = \tfrac{1+|\lambda|}{2}$ so that $\lambda \in \D_\rho$.  
By Lemma~\ref{lem:edgeremovebound}, there exists $M>0$ such that  
$|P_{G,e}(z)| \le M$ for all $z \in \partial \D_\rho$ and all $G,e$.  
Applying Lemma~\ref{lem:truncation-error} to $P_{G,e}$ yields
\[
  \bigl|P_{G,e}(\lambda) - P_{G,e}^{[k]}(\lambda)\bigr|
  \;\le\; C\, r^{-k}, \qquad r = \rho/|\lambda| > 1,
\]
for some constant $C$ independent of $G,e$. Moreover, Lemma~\ref{lem:edgeremovebound} with $S=\{\lambda\}$ provides a uniform lower bound $b>0$ such that 
$|P_{G,e}(\lambda)| \ge b$ for all $G,e$.  
Therefore,
\[
  \frac{\bigl|P_{G,e}(\lambda) - P_{G,e}^{[k]}(\lambda)\bigr|}
       {\bigl|P_{G,e}(\lambda)\bigr|}
  \;\le\; \frac{C}{b}\, r^{-k}.
\]
Choosing $
  k = \left\lceil\frac{\log\!\bigl((4mC)/(\varepsilon b)\bigr)}{\log r} \right\rceil $
guarantees the desired bound.  
Since $C,b$ and $r$ are independent of $G,e$, this choice satisfies 
$k=O(\log(m/\varepsilon))$ uniformly over all $G,e$.
\end{proof}


\subsection{Computing the truncated series via Weitz's tree reduction}

For simplicity, we omit the parameters $(\beta,\lambda)$ in the following.
Suppose $e_i=(u,v)$.
By the definition of the Ising model, we have
\begin{align*}
      P_{G_i,e_i} = \frac{Z_{G_{i-1}}}{Z_{G_i}} 
    = &\frac{
        Z_{G_{i-1},u,v}^{+,+} + 
        Z_{G_{i-1},u,v}^{-,-} +
        Z_{G_{i-1},u,v}^{-,+} +
        Z_{G_{i-1},u,v}^{+,-}
    }{
        Z_{G_{i},u,v}^{+,+}  + 
        Z_{G_{i},u,v}^{-,-} +
        Z_{G_{i},u,v}^{-,+} +
        Z_{G_{i},u,v}^{+,-}
    } \\ 
    =  &\frac{
        \dfrac{1}{\beta}Z_{G_{i},u,v}^{+,+} + 
        \dfrac{1}{\beta}Z_{G_{i},u,v}^{-,-} +
        Z_{G_{i},u,v}^{-,+} +
        Z_{G_{i},u,v}^{+,-}
    }{
        Z_{G_{i},u,v}^{+,+}  + 
        Z_{G_{i},u,v}^{-,-} +
        Z_{G_{i},u,v}^{-,+} +
        Z_{G_{i},u,v}^{+,-}
    } \\
    = & 1 + \left(1-\frac{1}{\beta}\right)\frac{Z_{G_{i},u,v}^{+,+} + Z_{G_{i},u,v}^{-,-}}{Z_{G_{i},u,v}^{+,+}  + 
        Z_{G_{i},u,v}^{-,-} +
        Z_{G_{i},u,v}^{-,+} +
        Z_{G_{i},u,v}^{+,-}} \\
    = & 1 + \left(1-\frac{1}{\beta}\right)
    \frac{R_{G_i,v}^{u^+}R_{G_i,u}^{v^-}+1}
    {R_{G_i,v}^{u^+}R_{G_i,u}^{v^-}+1 + R_{G_i,v}^{u^-}+R_{G_i,u}^{v^-}}.
\end{align*}
The second line holds because if $u$ and $v$ have the same spin, 
the only extra contribution in $Z_{G_i}$ (compared with $Z_{G_{i-1}}$) 
is the edge $e_i=(u,v)$, which contributes a factor of $\beta$.  
The last line follows by dividing numerator and denominator by $Z_{G_i,u,v}^{-,-}$ 
and substituting the ratio identities
\[
\frac{Z_{G_i,u,v}^{+,+}}{Z_{G_i,u,v}^{-,-}}
= \frac{Z_{G_i,u,v}^{+,+}}{Z_{G_i,u,v}^{+,-}}
  \cdot \frac{Z_{G_i,u,v}^{+,-}}{Z_{G_i,u,v}^{-,-}}
= R_{G_i,v}^{u^+} R_{G_i,u}^{v^-},\quad
\frac{Z_{G_i,u,v}^{+,-}}{Z_{G_i,u,v}^{-,-}} = R_{G_i,u}^{v^-},\quad
\frac{Z_{G_i,u,v}^{-,+}}{Z_{G_i,u,v}^{-,-}} = R_{G_i,v}^{u^-}.
\]

To calculate $P_{G_i,e_i}(\lambda)^{[k]}$, it suffices to calculate the ratios $R_{G_i,v}^{u^+}(\lambda)^{[k]}$, $R_{G_i,u}^{v^-}(\lambda)^{[k]}$
and $R_{G_i,v}^{u^-}(\lambda)^{[k]}$. This can be done by the tree recursion on the self-avoiding walk tree $T_{\text{SAW}}(G_i^{u^+},v)$, $T_{\text{SAW}}(G_i^{v^-},u)$ and $T_{\text{SAW}}(G_i^{u^-},v)$ respectively.
Recall the tree recursion formula
\[
  R_{T,r}(\lambda)
  = \lambda \prod_{i=1}^d
    \frac{\beta\, R_{T_i,v_i}(\lambda) + 1}{R_{T_i,v_i}(\lambda) + \beta}.
\]
To compute \(R_{T,r}(\lambda)^{[k]}\), it suffices to compute the truncated series of \(\prod_{i=1}^{d}\frac{\beta\,R_{T_i,v_i}(\lambda)+1}{R_{T_i,v_i}(\lambda)+\beta}\) up to degree \(k-1\); hence, for each child \(v_i\) we require \(R_{T_i,v_i}(\lambda)^{[k-1]}\).
Therefore \(R_{T,r}(\lambda)^{[k]}\) can be obtained by traversing the truncated self-avoiding walk tree to depth \(k\) and, for every node at depth \(i\in[0,k]\), computing the truncated series of its ratio to degree \(k-i\).
Suppose $T$ has maximum degree $\Delta$. At each node, we multiply at most \(\Delta\) series and perform one division, which costs \(O(\Delta\,k\log k)\) using FFT-based series arithmetic.
The truncated SAW tree to depth \(k\) has \(O(\Delta^{k})\) nodes.
Hence the total running time is \(O(\Delta^{k}\cdot\Delta\,k\log k)=O(\Delta^{k+1}k\log k)\).

If $G$ has maximum degree $\Delta$, then the self-avoiding walk tree $T_{\text{SAW}}(G_i^{u^+},v)$, $T_{\text{SAW}}(G_i^{v^-},u)$ and $T_{\text{SAW}}(G_i^{u^-},v)$ also have maximum degree $\Delta$. Thus, the time complexity for computing $P_{G_i,e_i}(\lambda)^{[k]}$ is $O(k \log k \Delta^{k+1})$.

\subsection{Approximation and running time}

\begin{theorem}\label{thm:weitz-alg}
Fix $\beta > 1$ and $\lambda \in \D$.  
There exists a deterministic algorithm that, given a graph $G=(V,E)$ with 
$m = |E|$ and maximum degree $\Delta$, and an accuracy parameter 
$\varepsilon \in (0,1)$, computes an approximation $\hat Z$ in time 
$ \left(\frac{m}{\varepsilon}\right)^{O(\log \Delta)} $
such that
\[
  \left| \frac{Z_G(\beta,\lambda) - \hat Z}{Z_G(\beta,\lambda)} \right|
  \le \varepsilon.
\]
\end{theorem}

\begin{proof}
Choose $k=O(\log(m/\varepsilon))$ as in \Cref{lem:truncation-relative} and set
$\hat Z=(1+\lambda)^{|V|}/\prod_{i=1}^m P_{G_i,e_i}^{[k]}(\lambda)$.  
For each $i$, let 
$
\delta_i \;=\; \frac{P_{G_i,e_i}(\lambda)}{P_{G_i,e_i}^{[k]}(\lambda)}-1$,
thus $|\delta_i| =
\left| \frac{{P_{G_i,e_i}^{[k]}(\lambda)}/{P_{G_i,e_i}(\lambda)}-1}{{P_{G_i,e_i}^{[k]}(\lambda)}/{P_{G_i,e_i}(\lambda)}}\right| \le \frac{\frac{\varepsilon}{4m}}{1-\frac{\varepsilon}{4m}} \leq \frac{\varepsilon}{2m}
$.
Then we have
\[
  \left| \frac{Z_G(\beta,\lambda)-\hat Z}{Z_G(\beta,\lambda)} \right|
  =\left|1-\prod_{i=1}^m(1+\delta_i)\right|
  \le \exp\!\left(\sum_{i=1}^m |\delta_i|\right)-1
  \le e^{\varepsilon/2}-1 \;\le\; \varepsilon.
\]
By the SAW-tree recursion and FFT-based truncated series arithmetic, each
$P_{G_i,e_i}^{[k]}(\lambda)$ is computable in $O(k\log k\,\Delta^{k+1})$ time; over all $m$ edges the total time is
\[
O(mk\log k\,\Delta^{k+1})
\;=\; \left(\frac{m}{\varepsilon}\right)^{O(\log \Delta)}. \qedhere
\]
\end{proof}

Our algorithm shows that SSM is unnecessary for a Weitz-type FPTAS, as we do not rely on SSM for the marginal probability of a vertex being assigned a particular spin. 
Instead, it is crucial for our algorithm to replace marginal probabilities of vertices by edge-deletion ratios, which eliminates the need for SSM. 
In the next section, we show that, even though  the standard notion of SSM does not hold, 
by further proving the local dependence of coefficients (LDC), a combinatorial property independent of zero-freeness, we can indeed establish a new form of SSM for edge deletion ratios.
Thus, zero-freeness alone gives Weitz-type FPTASes, while zero-freeness plus LDC gives new forms of SSM.

\section{SSM for the Random Cluster Model}\label{sec:edge-ssm}

SSM typically refers to the property that differences in conditional marginal probabilities at a given vertex exhibit exponential decay with respect to the distance of the disagreement condition in the Gibbs distribution.  

If we ignore the probabilistic meaning of $P_{G,v}$, arithmetically, it is just a ratio of two partition functions conditioning on different partial configurations. 
Such a ratio can be extended to a much more general setting. 
 For a partition function $Z_G(\vbeta, \vlambda)$ viewed as a multivariate function on edge activities $(\beta_e)_{e\in E}$ and vertex external fields $(\lambda_v)_{v\in V}$, and a partial evaluation $m(V', E'):(\beta_e)_{e\in E'}\rightarrow [1,\infty), (\lambda_v)_{v\in V'} \rightarrow \D$ (i.e., substituting specific values for variables $(\beta_e)_{e\in E'}$ and $(\lambda_v)_{v\in V'}$),
we consider the function $$Z^{m(V', E')}_G((\beta_e)_{e\in E\backslash E'}, (\lambda_v)_{v\in {V\backslash  V'}})$$ where the values of $(\lambda_v)_{v\in V'}$ and $(\beta_e)_{e\in E'}$ are assigned by $m(V', E')$.
When context is clear, we may omit the subscript ${e\in E\backslash E'}$ and ${v\in {V\backslash  V'}}$ in $Z^{m(V', E')}_G$.
Some particular partial evaluations have special meanings. 
For example, the assignment $m(u):\lambda_u\rightarrow 0$ that assigns the external field $\lambda_u$ of a particular vertex $u\in V$ to $0$ gives the function $Z^{m(u)}_G=Z^{-}_{G,u}$ which is the partition  function of the Ising model on the graph $G$ with a pinned vertex $u$ to the $-$ spin. 
Also, the assignment $m(e):\beta_e\rightarrow 1$ that assigns the edge activity $\beta_e$ of a particular edge $e\in E$ to $1$ gives the function $Z^{m(e)}_G=Z_{G-e}$ which is the partition  function of the Ising model on the graph $G-e$, i.e., the graph obtained from $G$ by removing the edge $e$.

In this paper, we focus on the partial evaluation $m(\emptyset, E')$ that only assigns values to edge activities for edges in $E'$.
For simplicity, we write $m(\emptyset, E')$ as $m(E')$. 
Then, as an extension of the marginal probability $P_{G, v}$,  we can define the ratio $P_{G,m(E')}(\vbeta,\vlambda) = {Z_{G}^{m(E')}(\vbeta,\vlambda)} / {Z_G(\vbeta,\vlambda)}$ for any partial evaluation $m(E')$.  
Moreover, we can define the ratio conditioning on a pre-specified partial evaluation $m_1(E_1)$ by
$  P_{G,m(E')}^{m_1(E_1)}(\vbeta,\vlambda)
    = {Z_{G}^{m_1(E_1),m(E')}(\vbeta,\vlambda)} / {Z_{G}^{m_1(E_1)}(\vbeta,\vlambda)}$ for partial evaluation $m(E')$ satisfying $E'\cap E_1 =\emptyset$.
If context is clear, we may omit the  arguments $(\vbeta,\vlambda)$ and the specification of edge sets $E_1$ and $E'$, and write $ P_{G,m(E')}^{m_1(E_1)}(\vbeta,\vlambda)$ as $P_{G,m}^{m_1}$ for simplicity.

With these notations in hand, we are able to define the generalized form of edge-type SSM.

\begin{definition}[Generalized edge-SSM]
	Let $\G$ be a family of graphs with parameters 
	$(\vbeta,\vlambda)$ and $C_2 \geq C_1 \geq 0$ be constants. 
	The Ising model defined on $\G$ is said to satisfy generalized edge-type strong spatial mixing (GE-SSM) with exponential rate $r>1$ if
	there exists a constant $C$ such that for any $G=(V,E) \in \G$, 
	any $e\in E$ with $m = \{ \beta_e \to \beta_e' \}$ where $\frac{\beta_e'}{\beta_e} \in [C_1,C_2]$ and
	any partial evaluation $m_1,m_2$ defined on $A,B \subseteq E\backslash\{v\}$ respectively, then
	
	\[	
		\abs{P_{G,m}^{m_1} - P_{G,m}^{m_2}} 
		\leq Cr^{-d_G(e, m_1 \neq m_2)}.
	\]  
%
	    Here, we denote $(m_1 \neq m_2) = (A\backslash B) \cup (B\backslash A) \cup \{f\in A\cap B: m_1(f) \neq m_2(f)\}$, 
    which is the set of edges where $m_1$ and $m_2$ differ. The quantity $d_G(e,m_1 \neq m_2)$ is the shortest distance from any endpoint of $e$ to any endpoint of an edge in $m_1 \neq m_2$.
\end{definition}

If we restrict the partial evaluations $m(E')$ and $m_1(E_1)$ to assigning edge activities only to the value $1$, then $ P_{G,m(E')}^{m_1(E_1)}=\frac{Z^{m(E'),m_1(E_1)}_{G}}{Z^{m_1(E_1)}_G}=\frac{Z_{G-E'-E_1}}{Z_{G-E_1}}$.
We define $P_{G,e}=\frac{Z_{G-e}}{Z_G}$.
Then, as a special form of GE-SSM, we define the following  edge-deletion form of SSM.

\begin{definition}[Edge-deletion SSM]
    Let $\G$ be a family of graphs with parameters 
	$(\vbeta,\vlambda)$. The Ising model defined on $\G$ is said to satisfy edge-deletion SSM
	with exponential rate $r>1$ if
	there exists a constant $C$ such that for any $G=(V,E) \in \G$, edge $e\in E$, sets of edge $A,B\subseteq E\backslash e$, then
	\[
		\abs{P_{G-A,e} - P_{G-B,e}} \leq Cr^{-d_G(e,A\neq B)}.
	\] 
\end{definition}

Indeed, we establish the GE-SSM result for the Ising model, as stated below.
\begin{theorem}\label{thm:eSSM}
    Fix constants $\delta \in (0,1)$ and $C_2 \geq C_1 \geq 0$. Then there exist constants $C>0$ and $r>1$ such that
    for all graphs $G=(V,E)$ with parameters $\vbeta \in [1,\infty)^E$ and $\vlambda \in (1-\delta)\D^V $, 
    and for any edge $e \in E$ and sets $A,B \subseteq E \backslash \{e\}$, the following holds. 
    Define the partial evaluation:
    \[ 
    m = \{\beta_e \ra \beta_e'\}, \quad 
    m_1 = \{\beta_f \ra \beta_f^A\}_{f\in A}, \quad
    m_2 = \{\beta_f \ra \beta_f^B\}_{f\in B}
    \]
    where $\beta_e' \in [1,\infty)$, $\beta_f^A \in [1,\infty]$ for all $f\in A$, $\beta_f^B \in [1,\infty]$ for all $f\in B$ and  
    $\frac{\beta_e'}{\beta_e} \in [C_1,C_2]$, we have
    \[
        \left|P_{G,m}^{m_1}(\vbeta,\vlambda)-P_{G,m}^{m_2}(\vbeta,\vlambda)\right|\leq Cr^{-d_G(e,m_1 \neq m_2)}.
    \]
\end{theorem}

\begin{remark}\label{rmk:finite_enough}
This theorem differs slightly from the definition of GE-SSM, as we allow the conditional partial evaluations to take the value
 $\infty$. However, this can be well-defined by taking limits of the corresponding ratios. 
Indeed, the Lee--Yang theorem ensure the ratios $P_{G,m}^{m_1}(\vbeta,\vlambda)$ and $P_{G,m}^{m_2}(\vbeta,\vlambda)$ is well-defined when for $\beta_f^A \in [1,\infty)$ for $f \in A$ and $\beta_f^B \in [1,\infty)$ for $f \in B$.
Once the theorem is established in this setting,
one can take the appropriate limits to extend its validity even when $\beta_f^A$ or $\beta_f^B$ approaches $\infty$.
\end{remark}

If set $\beta_e'=1$, then $\frac{\beta_e'}{\beta_e} \in [0,1]$ always holds.
As a corollary, the edge-deletion SSM holds.

\begin{corollary}\label{cor:edge-ssm}
Fix $\delta \in (0,1)$.
For any graph $G=(V,E)$ with $\vbeta \in [1,\infty)^E$ and $\vlambda \in (1-\delta)\D^V$, the edge-deletion SSM holds. 
\end{corollary}

Such an edge-type SSM does not have an explicit probabilistic meaning in the Ising model. However, through the relationship between the Ising model and the random cluster model, we found that it can be interpreted as the standard SSM in the random cluster model.

%
%

\subsection{LDC framework}

For two complex functions $f(z)$ and $g(z)$ analytic near $z_0$, we denote by $(z-z_0)^k\mid f(z)-g(z)$ the property that their Taylor series expansions, 
\[ 
	f(z)=\sum_{i=0}^{\infty}a_i (z-z_0)^i \quad\text{and}\quad g(z)=\sum_{i=0}^{\infty}b_i (z-z_0)^i
\]
satisfy $a_i=b_i$ for $0\leq i \leq k-1$.

The following lemma is a key tool in establishing SSM from zero-freeness, as used in \cite{regts2023absence,shao2024zero}. It also follows as a consequence of \Cref{lem:truncation-error}.

\begin{lemma}\label{lem:bound strip}
Let $f(z)$ and $g(z)$ be two analytic functions on some complex neighborhood $U$ of $z_0$. 
Suppose that the $(z-z_0)^n \mid f(z) - g(z)$.
Also, suppose that there exists an $M>0$ such that both $|f(z)|\le M$ and $|g(z)|\le M$ on some circle $\partial\mathbb{D}_{\rho}(z_0)\subseteq U$ $(\rho>0)$.
Then for every $z\in \mathbb{D}_{\rho}(z_0)$, we have 
\[\label{eq:bound 1}
    |f(z)-g(z)|\leq \frac{2M}{\rho(r-1)r^{n-1}}, \quad \text{with}\quad r=\dfrac{\rho}{|z-z_0|}>1.
\]
\end{lemma}

In \cite{shao2024zero}, Shao and Ye introduce the concept of local dependence of coefficients (LDC),  
which is implicitly used in \cite{regts2023absence}. To establish edge-type SSM for the Ising model,  
we introduce LDC below.
\begin{definition}[LDC]\label{def:LDC}
	We say that the Ising model satisfies LDC if for all
    graphs $G=(V,E)$ with parameters $\vbeta \in [1,\infty)^E$ and $\lambda \in \D$, the following holds.
    For an edge $e \in E$ and subsets $A,B \subseteq E \backslash \{e\}$, define the partial evaluations:
    \[
    	m = \{\beta_e \to \beta_e'\}, \quad
		m_1 = \{\beta_f \to \beta_f^A\}_{f\in A}, \quad
		m_2 = \{\beta_f \to \beta_f^B\}_{f\in B}
    \]
    where the modified parameters satisfy $\beta_e' \in [1,\infty)$, $\beta_f^A \in [1,\infty)$ for $f\in A$ and $\beta_f^B \in [1,\infty)$ for $f\in B$. 
    It holds that
    \[
        \lambda ^ {d_G(e,m_1 \neq m_2) + 1} \mid 
        \Bigr( P_{G,m}^{m_1}(\vbeta,\lambda)-P_{G,m}^{m_2}(\vbeta,\lambda)
        \Bigl).
    \]
\end{definition}

To address the non-uniform external field, we prove a slightly modified form of LDC. Once we have the LDC and a uniform bound, we can establish the edge SSM.

\begin{definition}[LDC]\label{def:modified-LDC}
	We say that the Ising model satisfies LDC if for all
    graphs $G=(V,E)$ with parameters $\vbeta \in [1,\infty)^E$ and $\vlambda \in \D^V$, the following holds.
    For an edge $e \in E$ and subsets $A,B \subseteq E \backslash \{e\}$, define the partial evaluations:
    \[
    	m = \{\beta_e \to \beta_e'\}, \quad
		m_1 = \{\beta_f \to \beta_f^A\}_{f\in A}, \quad
		m_2 = \{\beta_f \to \beta_f^B\}_{f\in B}
    \]
    where the modified parameters satisfy $\beta_e' \in [1,\infty)$, $\beta_f^A \in [1,\infty)$ for $f\in A$ and $\beta_f^B \in [1,\infty)$ for $f\in B$.
    It holds that
    \[
        z ^ {d_G(e,m_1 \neq m_2) + 1} \mid  
        \Bigl( P_{G,m}^{m_1}(\vbeta,\vlambda z)-P_{G,m}^{m_2}(\vbeta,\vlambda z) \Bigr).
    \]
\end{definition}

\subsection{Divisibility Relation via a Combinatorial Bijection}\label{sec:div}

We establish a divisibility relation that implies the LDC.

\begin{lemma}\label{lem:LDC}
    Let $G=(V,E)$ be a graph with parameters $(\vbeta,\vlambda)$ where $\vbeta \in [1,\infty)^E$ and $\vlambda \in \D^V$,
     Let $A,B \subseteq E$ be two disjoint edge sets, define the partial evaluations:
    \[
		m_1 = \{\beta_f \to \beta_f^A\}_{f\in A}, \quad
		m_2 = \{\beta_f \to \beta_f^B\}_{f\in B}
    \]
    where the modified parameters satisfy $\beta_f^A \in [1,\infty)$ for $f\in A$ and $\beta_f^B \in [1,\infty)$ for $f\in B$. Then
    \[
        z^{d_G(A,B)+1} \mid 
        \Bigl(
            Z_G(\vbeta,\vlambda z) Z_{G}^{m_1,m_2}(\vbeta,\vlambda z)-Z_{G}^{m_1}(\vbeta,\vlambda z)Z_{G}^{m_2}(\vbeta,\vlambda z)
        \Bigr)
    \]
    where $d_G(A,B) = \min_{e_1 \in A,e_2 \in B} d_G(e_1,e_2)$.
\end{lemma}

\begin{proof}
    For simplicity, we omit $(\vbeta,\vlambda z)$ in the notation.
    Let $\mathcal{S}= V \rightarrow\{+,-\}$, then
    \begin{align*}
          & Z_G Z_{G}^{m_1,m_2}-Z_{G}^{m_1}Z_{G}^{m_2}            \\
        = &
        \sum_{\sigma \in \mathcal{S}} w_G(\sigma)
        \sum_{\sigma \in \mathcal{S}} w_{G}^{m_1,m_2}(\sigma)
        -
        \sum_{\sigma \in \mathcal{S}} w_{G}^{m_1}(\sigma)
        \sum_{\sigma \in \mathcal{S}} w_{G}^{m_2}(\sigma) \\
        = &
        \sum_{\substack {(\sigma_1,\sigma_2) \in            \\ (\mathcal{S} \times \mathcal{S})}}
        w_G(\sigma_1) w_{G}^{m_1,m_2}(\sigma_2)
        -
        \sum_{\substack {(\sigma_3,\sigma_4) \in            \\ (\mathcal{S} \times \mathcal{S})}}
        w_{G}^{m_1}(\sigma_3)w_{G}^{m_2}(\sigma_4)
    \end{align*}

    Let $R = \{(\sigma, \tau) \in \mathcal{S} \times \mathcal{S} : n_+(\sigma) + n_+(\tau) < d(A,B) + 1 \}$, where 
    $n_+(\sigma)$ is the number of vertices with $+$ spin in $\sigma$. 
    We will show that there exists an automorphism $f$ on $R$ such that if 
    $(\sigma_3,\sigma_4) = f(\sigma_1,\sigma_2)$, then $w_G(\sigma_1) w_{G}^{m_1,m_2}(\sigma_2) = w_{G}^{m_1}(\sigma_3)w_{G}^{m_2}(\sigma_4)$.

    Let $(\sigma_1, \sigma_2) \in R$, consider the subgraph
    $H = (V, E_+(\sigma_1 | \sigma_2))$, where $\sigma_1 \mid \sigma_2$ denotes the logical OR, interpreting $+$ as true.
    Since $n_+(\sigma_1) + n_+(\sigma_2) < d(A,B) + 1$, there are no paths connecting any edge between $A$ and $B$ in $H$.
    Let $S$ be the minimal vertex set containing all connected components of $H$ that intersect with $G[A]$ and $T = V\backslash S$.
    Swap the part at $T$ of $\sigma_1$ and $\sigma_2$, write it as $(\sigma_3,\sigma_4) =  (\sigma_2|_S \cup  \sigma_1|_T ,\sigma_1|_S \cup  \sigma_2|_T)$.
    Obviously, $(\sigma_3,\sigma_4) \in R$ and the process is reversible (note $\sigma_3 | \sigma_4 = \sigma_1 | \sigma_2$,
    which is unchanged in the process), thus $f$ is an automorphism.

    Since there are no $(+,+)$ edges between $S$ and $T$ for $\sigma_1 | \sigma_2 = \sigma_3 | \sigma_4$, it follows that
    there are no $(+,+)$ edges between $S$ and $T$ for $\sigma_1,\sigma_2,\sigma_3$ and $\sigma_4$. 
    For an edge $e = (u,v) \in E$ between $S$ and $T$, define $s(e,\sigma) = \mathbbold{1}[e \in E_-(\sigma)]$.
     Recalling that $e$ cannot be a $(+,+)$ edge in any $\sigma_i (i=1,2,3,4)$, we obtain 
     $s(e,\sigma) = 1 -  \mathbbold{1}[\sigma(u)=+] -  \mathbbold{1}[\sigma(v)=+]$.
    Moreover, note that $\mathbbold{1}[\sigma_1(u)=+] + \mathbbold{1}[\sigma_2(u)=+] = \mathbbold{1}[\sigma_3(u)=+] + \mathbbold{1}[\sigma_4(u)=+]$ 
     and similarly for $v$. It follows that $s(e,\sigma_1)+s(e,\sigma_2) = s(e,\sigma_3)+s(e,\sigma_4)$.

    Let $C = \{(u,v)\in E\mid u\in S, v\in T \}$ be the set of cut edges between $S$ and $T$. By the definition of $w(\cdot)$, we have
    \begin{align*}
          & w_G(\sigma_1)w_{G}^{m_1,m_2}(\sigma_2)             \\
        = &
        \prod_{e\in C}\beta_e^{s(e,\sigma_1)}
        w_{G[S]}(\sigma_1 |_S)w_{G[T]}(\sigma_1 |_T)
        \prod_{e\in C}\beta_e^{s(e,\sigma_2)}
        w_{G[S]}^{m_1}(\sigma_2 |_S)w_{G[T]}^{m_2}(\sigma_2 |_T) \\
        = &
        \prod_{e\in C}\beta_e^{s(e,\sigma_1)+s(e,\sigma_2)}
        w_{G[S]}^{m_1}(\sigma_2 |_S)w_{G[T]}(\sigma_1 |_T)
        w_{G[S]}(\sigma_1 |_S)w_{G[T]}^{m_2}(\sigma_2 |_T)     \\
        = &
        \prod_{e\in C}\beta_e^{s(e,\sigma_3)+s(e,\sigma_4)}
        w_{G[S]}^{m_1}(\sigma_3 |_S)w_{G[T]}(\sigma_3 |_T)
        w_{G[S]}(\sigma_4 |_S)w_{G[T]}^{m_2}(\sigma_4 |_T)     \\
        = &
        \prod_{e\in C}\beta_e^{s(e,\sigma_3)}
        w_{G[S]}^{m_1}(\sigma_3 |_S)w_{G[T]}(\sigma_3 |_T)
        \prod_{e\in C}\beta_e^{s(e,\sigma_4)}
        w_{G[S]}(\sigma_4 |_S)w_{G[T]}^{m_2}(\sigma_4 |_T)     \\
        = & w_{G}^{m_1}(\sigma_3)w_{G}^{m_2}(\sigma_4).
    \end{align*}
    Thus, the proof is complete. \qedhere

\end{proof}

\subsection{Generalized edge-SSM}\label{sec:essm}
\subsubsection{Edge-type LDC}

\begin{lemma}
    Let $G=(V,E)$ be a graph with parameters $\vbeta \in [1,\infty)^E$ and $\vlambda \in \D^V$. Let $e\in E$ and $A \subseteq E\backslash\{e\} $, 
    $m = \{ \beta_e \ra \beta_e'\}$ with $\beta_e' \geq 1$ and 
    $m_1 = \{ \beta_f \ra \beta_f' \mid f \in A \}$ where $\beta_f' \geq 1$ for all $f\in A$. Then the Taylor series of $P_{G,m}(\vbeta, \vlambda z)$ and $P_{G,m}^{m_1}(\vbeta, \vlambda z)$ 
    near $z = 0$ satisfy
    \[
        z^{d_G(e,A)+1} \mid \Bigl( P_{G,m}(\vbeta, \vlambda z) - P_{G,m}^{m_1}(\vbeta, \vlambda z) \Bigr).
    \]
\end{lemma}

\begin{proof}
    \begin{align*}
    P_{G,m}(\vbeta, \vlambda z) - P_{G,m}^{m_1}(\vbeta, \vlambda z) 
        =& \frac{Z_G^m(\vbeta, \vlambda z)}{Z_G(\vbeta, \vlambda z)} - \frac{Z_G^{m,m_1}(\vbeta, \vlambda z)}{Z_G^{m_1}(\vbeta, \vlambda z)}   \\
        =& \frac{Z_G^m(\vbeta, \vlambda z)Z_G^{m_1}(\vbeta, \vlambda z) - Z_G^{m,m_1}(\vbeta, \vlambda z)Z_G(\vbeta, \vlambda z)}{Z_G(\vbeta, \vlambda z)Z_G^{m_1}(\vbeta, \vlambda z)}.
    \end{align*}
    Clearly by the Lee--Yang theorem, $\frac{1}{Z_G(\vbeta, \vlambda z)Z_G^{m_1}(\vbeta, \vlambda z)}$ is analytic near $z = 0$.  
	Combining this with \Cref{lem:LDC}, we have  
	\[
		z^{d_G(e,A)+1} \mid \Bigl(P_{G,m}(\vbeta, \vlambda z) - P_{G,m}^{m_1}(\vbeta, \vlambda z)\Bigr).
	\]
\end{proof}

\begin{lemma}[LDC]\label{lem:edgeLDC}
    Let $G=(V,E)$ be a graph with parameters $\vbeta \in [1,\infty)^E$ and $\vlambda \in \D^V$. Let $e\in E$ and $A,B \subseteq E\backslash\{e\} $, and partial evaluations 
    \[
    	m = \{\beta_e \to \beta_e'\}, \quad
		m_1 = \{\beta_f \to \beta_f^A\}_{f\in A}, \quad
		m_2 = \{\beta_f \to \beta_f^B\}_{f\in B}
    \]
    where $\beta_e' \in [1,\infty)$, $\beta_f^A \in [1,\infty)$ for $f\in A$ and $\beta_f^B \in [1,\infty)$ for $f\in B$.
    Then the Taylor series of $P_{G,m}^{m_1}(\vbeta, \vlambda z)$ and $P_{G,m}^{m_2}(\vbeta, \vlambda z)$ near $z = 0$ satisfy
    \[
        z^{d_G(e,m_1\neq m_2)+1} \mid \Bigl(P_{G,m}^{m_1}(\vbeta, \vlambda z) - P_{G,m}^{m_2}(\vbeta, \vlambda z)\Bigr).
    \]
\end{lemma}
\begin{proof}
    Define $\vbeta'$ as $\vbeta$ after applying by $m_1 \cap m_2$, let $m_1' = m_1 \backslash m_2$ and $m_2' = m_2 \backslash m_1$, then 
    \begin{align*}
        P_{G,m}^{m_1}(\vbeta, \vlambda z) - P_{G,m}^{m_2}(\vbeta, \vlambda z) 
        =& P_{G,m}^{m_1'}(\vbeta', \vlambda z) - P_{G,m}^{m_2'}(\vbeta', \vlambda z) \\
        =& [ P_{G,m}^{m_1'}(\vbeta', \vlambda z) - P_{G,m}(\vbeta', \vlambda z) ] + [ P_{G,m}(\vbeta', \vlambda z) - P_{G,m}^{m_2'}(\vbeta', \vlambda z) ].
    \end{align*}

    By the previous lemma, we have $z^{d_G(e,m_1')+1} \mid \Bigl(P_{G,m}^{m_1'}(\vbeta', \vlambda z) - P_{G,m}(\vbeta', \vlambda z)\Bigr)$ and $z^{d_G(e,m_2')+1} \mid \Bigl(P_{G,m}(\vbeta', \vlambda z) - P_{G,m}^{m_2'}(\vbeta', \vlambda z)\Bigr)$. 
    Since $d_G(e, m_1 \neq m_2) = \min \{ d_G(e,m_1'), d_G(e,m_2') \}$, we are done. 
\end{proof}

\subsubsection{Uniform bound of edge-type ratio}
We are ready to prove the edge-type ratio avoids $0$ and $1$.
\begin{lemma}
    \label{lem:edgeavoid01}
    Let $G=(V,E)$ be a graph,
    with parameters
    $\vbeta \in [1,\infty)^E$ and $\vlambda \in \D^{V}$,
    edge $e \in E$, if $\beta_e' \geq 1$ and $\beta_e' \neq \beta_e$, 
    then $P_{G,\{\beta_e \ra \beta_e'\}}(\vbeta,\vlambda)$ 
    avoids $0$ and $1$.


\end{lemma}

\begin{proof}
    Since $\beta_e' \geq 1$, by the Lee--Yang theorem, it is trivial that $P_{G,\{\beta_e \ra \beta_e'\}}(\vbeta,\vlambda) \neq 0$. We prove the ratio avoids $1$.

    Let $e = (u,v)$, we have
    \begin{align*}
              & Z_{G}(\vbeta,\vlambda) - Z_{G}^{\{\beta_e \ra \beta_e'\}}(\vbeta,\vlambda)                                                                                               \\
        =     & Z_{G,u,v}^{+,+}(\vbeta,\vlambda) + Z_{G,u,v}^{-,-}(\vbeta,\vlambda) + Z_{G,u,v}^{+,-}(\vbeta,\vlambda) + Z_{G,u,v}^{-,+}(\vbeta,\vlambda)                                            \\
        \quad & - \frac{\beta_e'}{\beta_e} Z_{G,u,v}^{+,+}(\vbeta,\vlambda) -\frac{\beta_e'}{\beta_e} Z_{G,u,v}^{-,-}(\vbeta,\vlambda) - Z_{G,u,v}^{+,-}(\vbeta,\vlambda) - Z_{G,u,v}^{-,+}(\vbeta,\vlambda) \\
        =     & \frac{\beta_e - \beta_e'} {\beta_e} (Z_{G,u,v}^{+,+}(\vbeta,\vlambda) + Z_{G,u,v}^{-,-}(\vbeta,\vlambda)).
    \end{align*}

    Merge $u,v$ into a single vertex $w$ we get graph $G'=(V',E')$, set $\lambda_w = \lambda_u \lambda_v$, if parallel edges exist (i.e. $(u,x) \in E, (v,x) \in E$ for some $x\in V$), we merge them into a single edge 
    and set $\beta_{(w,x)} = \beta_{(u,x)}\beta_{(v,x)}$. Write the partition function of $G'$ with new parameters as $Z_{G'}(\vbeta',\vlambda')$.
    
    One can see $ Z_{G'}(\vbeta',\vlambda')= Z_{G',w}^+(\vbeta',\vlambda') +Z_{G',w}^-(\vbeta',\vlambda') = ( Z_{G,u,v}^{+,+}(\vbeta,\vlambda) + Z_{G,u,v}^{-,-}(\vbeta,\vlambda)) / \beta_e$. 
    Since $\vlambda' \in \D^{V'}$ and $\vbeta' \in [1,\infty)^{E'}$,
    by the Lee--Yang theorem, $Z_{G}(\vbeta,\vlambda) - Z_{G}^{\{\beta_e \ra \beta_e'\}}(\vbeta,\vlambda) = (\beta_e - \beta_e')Z_{G'}(\vbeta',\vlambda') \neq 0$. Thus the ratio avoids $1$.
\end{proof}


\begin{lemma}[uniform bound]\label{lem:edge_uniform_bound}
    Fix constants $\delta \in (0,1)$ and $C_2 \geq C_1 \geq 0$. Let $S$ be a compact subset of $\frac{1}{1-\delta}\D$.
    Then, there exists a constant $C>0$ such that for any graph $G=(V,E)$ with parameters 
    $\vbeta \in [1,\infty)^E$ and $\vlambda \in (1-\delta)\D^V$, for any $e\in E$, any $\beta'_e \geq 1$ with
     $ \frac{\beta_e'}{\beta_e} \in [C_1,C_2]$, we have
     $|P_{G,\{\beta_e \ra \beta_e'\}}(\vbeta,\vlambda z)| \leq C$ for all $z\in S$.
\end{lemma}

\begin{proof}
     Consider the family of functions $f(z)=P_{G,\{\beta_e \ra \beta'\}}(\vbeta,\vlambda z)$
	where $z$ is the variable. It's trivial when $\beta_e' = \beta_e$, the ratio is exactly $1$.
     So we only consider the family of ratio functions when $\beta_e' \neq \beta_e$. 
     By \Cref{lem:edgeavoid01}, $P_{G,\{\beta_e \ra \beta_e'\}}(\vbeta,\vlambda z)$ avoids
      $0$ and $1$ for all $z \in \frac{1}{1-\delta}\D$. Since 
     $P_{G,\{\beta_e \ra \beta_e'\}}(\vbeta,\vlambda \cdot 0) = \frac{\beta_e'}{\beta_{e}} \in [C_1,C_2] $ 
     is bounded, by \Cref{lem:montel-bound}, we obtain the upper bound.
\end{proof}


Now we are ready to prove edge-type SSM of the Ising model and then immediately deduce the SSM of the random cluster model.

\begin{proof}[Proof of \Cref{thm:eSSM}]
	By \Cref{lem:edgeLDC}, we have $z^{d_G(e,m_1\neq m_2)+1} \mid \Bigl( P_{G,m}^{m_1}(\vbeta, \vlambda z) - P_{G,m}^{m_2}(\vbeta, \vlambda z) \Bigr)$. Let $S = (1+\delta)\partial\D$, which is a compact subset of $\frac{1}{1-\delta}\D$.
	By \Cref{lem:edge_uniform_bound} we know that the ratio is uniformly bounded for $z \in S$. 
	Choosing $z = 1\in (1+\delta)\D$, we apply \Cref{lem:bound strip} to conclude the proof.
\end{proof}

\subsection{SSM for the random cluster model}

Let $G=(V,E)$ be a graph, and let $\vp \in [0,1]^E$ and $\vlambda \in [0,1]^V$ be parameters. The weight of a configuration $S\subseteq E$ in the (weighted) random cluster model is defined by:
\[
    w_{G,\vp,\vlambda}^{\text{RC}}(S) = \prod_{e\in S}p_e \prod_{e\in E\backslash S} \left(1-p_e\right) \prod_{C\in \kappa(V,S)}
    \left(1+\prod_{j \in C} \lambda_j\right),
\]
where $\kappa(V,S)$ denotes the set of connected components of graph $(V,S)$. 
The partition function of the random cluster model is given by
\[
Z^{\text{RC}}_{G}(\vp,\vlambda) = \sum_{S\subseteq E}w_{G,\vp,\vlambda}^{\text{RC}}(S).
\]
When \( \vlambda = \vone \), the weighted random cluster model reduces to the standard random cluster model for the Ising model without external field. The relationship between the Ising model with an external field and the random cluster model is given in the following lemma.

\begin{lemma}[{\cite[Proposition 2.1]{feng2023swendsen}}]\label{lem:rc-ratio}
Let $G = (V,E)$ be a graph, and let $ \vbeta \in [1,+\infty)^E $ and \( \vlambda \in [0,1]^V \) be parameters. Then,
\[
    Z_G^\Ising(\vbeta,\vlambda) = \left( \prod_{e\in E} \beta_e \right)  Z_G^\RC(\vp,\vlambda),
\]
where \( \vp = 1 - \vbeta^{-1} = (1-\beta_e^{-1})_{e\in E} \).
\end{lemma}

\begin{remark}\label{rmk:rc-ratio}
    Expressing it as $Z_G^\RC(\vp,\vlambda) = Z_G^\Ising(\vbeta,\vlambda) /  \prod_{e\in E} \beta_e$,
    we observe that setting $\beta_e = \infty$ is well-defined by taking the limit, which corresponds to setting $p_e=1$ in the random cluster model.
\end{remark}

When $\vp \in [0,1]^E$ and $\vlambda \in [0,1]^V$, RC model induces a distribution $\mu(\cdot)$ 
where $\mu(S)=w(S) /Z$ for $S\subseteq E$. Denote the marginal probability on an edge $e$ 
such that $e$ is picked and unpicked as $P_{G,e}^+(\vp,\vlambda)=Z_{G,e}^{+}/Z_G$ and 
$P_{G,e}^-(\vp,\vlambda)=Z_{G,e}^{-}/Z_G$ where $Z_{G,e}^{+} = \sum_{S\subseteq E, e\in S}  w(S)$
and $Z_{G,e}^{-} = \sum_{S\subseteq E\backslash e} w(S)$
respectively. We also define the partition function conditioning on a pre-described partial
configuration $\sigma_A$ ($A\subseteq E$, each edge in $A$ is pinned to be in or out the configurations, 
we use the notation $+$ and $-$ denoting in and out) denoted by 
\[
Z_G^{\sigma_A} = \sum_{\substack{S \subseteq E \\ S|_A = \sigma_A}} w_{G,\vp,\vlambda}^{\text{RC}}(S)
\]
and then the conditional marginal probabilities $e$ unpicked under condition $\sigma_A$ are defined by 
\[
	P_{G,e}^{\sigma_A} = \frac{Z_{G,e}^{\sigma_A,-}}{Z_G^{\sigma_A}}.
\]

\begin{definition}[SSM for the random cluster model]
    Let $\G$ be a family of graphs with parameters $(\vp,\vlambda)$. 
    The random cluster model defined on $\G$
	is said to satisfy strong spatial mixing with exponential rate $r>1$ if
	there exists a constant $C$ such that for any $G=(V,E) \in \G$, any edge $e \in V$, any partial configuration
	$\sigma_{\Lambda_1}$ and $\tau_{\Lambda_2}$ where $\Lambda_1,\Lambda_2 \subseteq E\backslash e$, we have
	\[	
		\abs{P_{G,e}^{\sigma_{\Lambda_1}}(\vp,\vlambda) - P_{G,e}^{\tau_{\Lambda_2}}(\vp,\vlambda)} 
		\leq Cr^{-d_G(e,\sigma_{\Lambda_1} \neq \tau_{\Lambda_2})}.
	\]
\end{definition}

\begin{lemma}\label{lem:rc=ising}
    The conditional marginal probability of edge $e$ under condition $\sigma_A$ for $A\subseteq E\backslash e$ in the random cluster model can be translated to the edge-type ratio in the Ising model as 

\[
	P_{G,e}^{\sigma_A} = \frac{Z_{G-e}^{\Ising,m(\sigma_A)}}{Z_{G}^{\Ising,m(\sigma_A)}}
\]
where $m(\sigma_A) = \{\beta_e \to 1 \mid \sigma_A(e) =-\} \cup \{\beta_e \to \infty \mid \sigma_A(e) =+\}$.

\end{lemma}

\begin{proof}
    
Pinning an edge $e$ picked or unpicked can also be understood via the modifying on the parameters, as stated in the
\[
 	Z_{G,e}^{\RC,+} = p_e Z_{G}^{\RC}(p_e=1) \quad \text{and} \quad Z_{G,e}^{\RC,-} = (1-p_e) Z_{G}^{\RC}(p_e=0).
\]
The corresponding modifying is setting $\beta_e = \infty$ and $\beta_e=1$ respectively. One can use the rule recursively, 
let $\vp^{\sigma_A}$ denote $\vp$ modified by setting $p_e = 1$ for $\sigma_A(e)=+$ and $p_e = 0$ for $\sigma_A(e)=-$ 
respectively.
\[
	P_{G,e}^{\sigma_A} = (1-p_e)\frac{Z_G(\vp^{\sigma_A},p_e=0)}{Z_G(\vp^{\sigma_A})}.
\]
Then by \Cref{lem:rc-ratio} and \Cref{rmk:rc-ratio},
\[
	P_{G,e}^{\sigma_A} = \frac{Z_{G-e}^{\Ising,m(\sigma_A)}}{Z_{G}^{\Ising,m(\sigma_A)}}. \qedhere
\]
\end{proof}

Thus, the GE-SSM or edge-deletion SSM of the Ising model will directly imply the SSM of the random cluster model.

%

\begin{theorem}[SSM for the random cluster model]\label{thm:ssm_rc}
    Fix a constant $\delta \in (0,1)$. For any graph $G=(V,E)$ 
    with parameters $\vp \in [0,1]^E$ and $\vlambda \in [0,1-\delta]^V$, 
   	SSM holds for the random cluster model.
\end{theorem}

\begin{proof}  
    Following the transformation between the Ising model and the random cluster model,  
    the result follows immediately from \Cref{thm:eSSM}.  
\end{proof}

\subsection{Optimal mixing time on lattice}

The mixing time result is a direct consequence of the SSM result in \Cref{thm:ssm_rc} and the framework in \cite{gheissari2024spatial}.

\subsubsection{Markov chain and mixing time}
Let $(X_t)_{t\in \N}$ be a Markov chain over a finite state space $\Omega$ with transition matrix $P$.  
$(X_t)_{t\in \N}$ is irreducible if for any $x,y\in \Omega$, there exists $t>0$ such that $P^t(x,y) > 0$.  
$(X_t)_{t\in \N}$ is aperiodic if for any $x\in \Omega$,  $\gcd\{t\in \N^+ \mid P^t(x,x) > 0\} = 1$.
A distribution $\mu$ over $\Omega$ is a stationary distribution of $(X_t)_{t\in \N}$ if $\mu P = \mu$.  
If the Markov chain is irreducible and aperiodic, then it has a unique stationary distribution.  
The total variation distance between two distributions $\mu,\nu$ on the same state space $\Omega$ is defined as  
$
d_{\mathrm{TV}}(\mu,\nu) = \max_{S\subseteq \Omega} \abs{\mu(S)-\nu(S)} = \frac{1}{2} \sum_{x \in \Omega} \abs{\mu(x)-\nu(x)}.
$  
Suppose $\mu$ is the stationary distribution of $(X_t)_{t\in \N}$. The mixing time of the chain is defined as  
\[
T_{\mathrm{mix}}(\epsilon) = \max_{x_0 \in \Omega} \min\{t \in \N \mid d_{\mathrm{TV}}(P^t(x_0,\cdot),\mu) < \epsilon\}.
\]  
By convention, the standard mixing time is defined as  $
T_{\mathrm{mix}} = T_{\mathrm{mix}}\left(\frac{1}{4}\right)$.

\subsubsection{Glauber dynamics for the random cluster model}

The Glauber dynamics for the random cluster model (FK dynamics) is defined as follows.
If the configuration at time $t$ is $\sigma$, then the configuration at time $t+1$ is obtained by
\begin{enumerate}
    \item Pick an edge $e \in E$ uniformly at random.
    \item Include $e$ in the new configuration with probability 
        $p(\sigma,e) = \frac{\mu(\sigma\cup\{e\})}{\mu(\sigma\cup\{e\}) + \mu(\sigma\backslash\{e\})}$, otherwise exclude it.
\end{enumerate}

This dynamics is irreducible and reversible with respect to the distribution $\mu(\cdot)$ induced by the random cluster model,
it coverages to the distribution $\pi(\cdot)$ no matter what the initial configuration is.
Write $p(\sigma, e)$ explicitly as follows. Suppose $e = (u, v)$ is an edge in the graph. If $e$ is not a cut edge in the configuration $\sigma \cup e$, then $p(\sigma, e) = p_e$. Otherwise, suppose $u$ and $v$ belong to distinct connected components $C_1$ and $C_2$ in $\sigma$, respectively. Let $x_1 = \prod_{v \in C_1} \lambda_v$ and $x_2 = \prod_{v \in C_2} \lambda_v$. Then,
\[
p(\sigma, e) = \frac{p_e (1 + x_1 x_2)}{p_e (1 + x_1 x_2) + (1 - p_e)(1 + x_1)(1 + x_2)}.
\]

If the parameters of the corresponding Ising model satisfy $\beta_e \in [\beta_{\min}, \beta_{\max}]$ and $\lambda_v \in [0, 1]$, where $1 < \beta_{\min} \leq \beta_{\max}$, then the following bounds hold: $p(\sigma, e) \leq p_e \leq 1 - \frac{1}{\beta_{\max}}$ and $p(\sigma, e) \geq \frac{p_e}{6} \geq \frac{1}{6} \left(1 - \frac{1}{\beta_{\min}}\right)$. Thus, $p(\sigma, e)$ is uniformly bounded away from $0$ and $1$ by a constant distance,
i.e., the minimum probability that an edge unchanged in an update step can be determined by $\beta_{\min}$ and
$\beta_{\max}$.

%

\subsubsection{Monotonicity and the grand coupling}
The grand coupling of the Glauber dynamics can be defined as follows. For a graph $G=(V,E)$, 
starting from a configuration $\omega$ and a boundary condition $\sigma_\Lambda$, 
the grand coupling is given by $\{ X_{t,\sigma_\Lambda} ^{\omega} \}$, indexed by the initial configuration $\omega$ (or a distribution) and the boundary condition $\sigma_\Lambda$. 
We assign a Poisson clock of rate $1$ to each edge $e\in E$. When the clock for an edge $e$ rings at time $t$, 
we sample a random variable $U_t$ uniformly from $[0,1]$. If $e \in \Lambda$, the configuration remains unchanged; otherwise, we update the configuration according to the Glauber dynamics: 
if $U_t \geq 1 - p(\sigma,e)$, we include $e$ in the configuration; otherwise, we exclude $e$. 

As in the standard random cluster model, the grand coupling of the weighted random cluster model is monotonic.
\begin{lemma}\label{lem:mono}\cite[Lemma 8.2]{feng2023swendsen}
Suppose $0 \leq p_e < 1$ for all $e\in E$ and $0 \leq \lambda_v \leq 1$ for all $v\in V$.
Then the grand coupling of the Glauber dynamics for the weighted random cluster model is monotonic.
\end{lemma}

The key of lemma is the following inequality, suppose $\sigma_1 \leq \sigma_2$, then 
$p(\sigma_1,e) \leq p(\sigma_2,e)$. The monotonic grand coupling also implies the monotonicity
of the Glauber dynamics.

%

\subsubsection{Strong spatial mixing implies optimal mixing time}
Follow the approach used for the standard random cluster model on lattices in \cite{gheissari2024spatial}.
They consider the continuous Glauber dynamics and thus need a constant minimum probability that an edge is unchanged, which can be realized 
by setting $\beta_{\min}$ and $\beta_{\max}$. Since the weighted random cluster model still exhibits the monotonicity property and the monotonic grand coupling, their proof also applies, showing that strong spatial mixing in the weighted random cluster model implies the optimal mixing time of the Glauber dynamics.

\begin{theorem}
Fix $1 < \beta_{\min} \leq \beta_{\max}$, $\delta \in (0,1)$, and $d \in \mathbb{N}$.  
For any subgraph $G=(V,E)$ of the infinite $d$-dimensional lattice, with parameters $\vbeta \in [\beta_{\min},\beta_{\max}]^E$ and $\vlambda \in [0,1-\delta]^V$, the mixing time of the Glauber dynamics for the corresponding random-cluster representation of the Ising model is $O(m \log m)$, where $m = |E|$.
\end{theorem}

\section{LDC and SSM for Other Models}\label{sec:other-applications}

\subsection{SSM for the hypergraph independence polynomial}\label{sec:hyper-ind}

A hypergraph $H=(V,E)$ is a set of vertices $V$ along with a set of edges, where each edge is a nonempty subset of $V$.
The degree of a vertex $v$ in a hypergraph is the number of edges containing $v$.
An independent set in $H$ is a set of vertices $I\subseteq V$ such that no edge in $E$ is a subset of $I$. Let 
$\mathcal{I}$ be the set of all independent sets in $H$, then the independence polynomial of $H$ is defined as
\[
    Z_H(\lambda) = \sum_{I\in \mathcal{I}} \lambda^{|I|}.
\]


We continue using the notations $+$ and $-$ to denote the vertex being in and out of the independent set, respectively.
A partial configuration $\sigma_\Lambda$ is feasible if it is an independent set. We say $v$ is proper to 
$\sigma_\Lambda$ if $v \notin \Lambda$ and $\sigma_{\Lambda \cup \{v\}}^{+}$ is feasible. 

We need to define some operations on the hypergraph $H = (V,E)$.

\begin{enumerate}
    \item {\bf Induced sub-hypergraph.} 
    For $\Lambda \subseteq V$, denote $H_{\Lambda} = (\Lambda, \{e \cap \Lambda : e \in E, e \cap \Lambda \neq \varnothing \})$.

    \item {\bf Mild vertex deletion.}
    For $v \in V$, denote $H \ominus v = H_{V \backslash \{v\}}$. For $S \subseteq V$, 
    denote $H \ominus S = H_{V \backslash S}$.

    \item {\bf Total vertex deletion.}
    For $v \in V$, denote $H \backslash v = (V \backslash \{v\}, \{e \in E : v \notin e\})$. For 
    $S \subseteq V$, denote $H \backslash S = (V \backslash S, \{e \in E : e \cap S = \varnothing\})$.
\end{enumerate}

Note that $H \ominus v$ retains as many edges as possible while $H \backslash v$ removes all edges containing $v$.
In fact, we have the following relations \cite{trinks2016survey}: 
\[
    Z_{H,v}^+(\lambda) = \lambda Z_{H \ominus v}(\lambda) \quad \text{and} \quad Z_{H,v}^-(\lambda) = Z_{H \backslash v}(\lambda).
\]

Note $H \ominus v$ reverse the edges containing $v$ as possible while $H\backslash v$ remove all edges containing $v$.
Then one can see that $Z_{H,v}^+(\lambda) = \lambda Z_{H \ominus v}(\lambda)$ 
and $Z_{H,v}^-(\lambda) = Z_{H\backslash v}(\lambda)$.
The marginal probability that $v$ is in an independent set is defined as 
\[
    P_{H,v}(\lambda) =  \frac{Z_{H,v}^{+}(\lambda)}{Z_H(\lambda)} = \frac{\lambda Z_{H \ominus v}(\lambda)}{Z_H(\lambda)}.
\]
Similarly, the conditional probability that $v$ is in an independent set, given a partial configuration $\sigma_\Lambda$, is defined as
\[
    P_{H,v}^{\sigma_\Lambda}(\lambda) = \frac{Z_{H,v}^{\sigma_\Lambda,+}(\lambda)}{Z_{H}^{\sigma_\Lambda}(\lambda)}.
\]

For a partial configuration $\sigma_\Lambda$ that pins 
vertices in $\Lambda^+$ to $+$ and vertices in $\Lambda^-$ to $-$, where $(\Lambda^+, \Lambda^-)$ is a partition of $\Lambda$, 
denote $H[\sigma_\Lambda] = (H \ominus \Lambda^+) \backslash \Lambda^-$. 
Then we have the following identity:

\[
    P_{H,v}^{\sigma_\Lambda}(\lambda) = P_{H[\sigma_\Lambda],v}(\lambda),
\]
which allows us to analyze the ratio without the partial configuration. This identity can be verified 
by comparing each independent set in $H$ with $\sigma_\Lambda$ to those in $H[\sigma_\Lambda]$.


Denote $\lambda_c(\Delta) = \frac{(\Delta-1)^{\Delta-1}}{(\Delta-2)^{\Delta}}$ and 
$\lambda_s(\Delta) = \frac{(\Delta-1)^{\Delta-1}}{\Delta^{\Delta}}$.  
We directly state the new strong spatial mixing result for the hypergraph independence polynomial as a theorem,
which is stronger than the definition in \cite{bezakova2019approximation}.

\begin{theorem}[\Cref{thm:hyper-ssm-intro} restated]\label{thm:hyper-ssm}
    Fix $\Delta \geq 3$ and $\lambda \in \D_{\lambda_s(\Delta)} \cup (0,\lambda_c(\Delta))$. There
     exist constants $C>0$ and $r>1$ such that for any hypergraph $H=(V,E)$ with maximum degree at most $\Delta$, any 
     two feasible partial configurations $\sigma_{\Lambda_1}$ and $\tau_{\Lambda_2}$ where $\Lambda_1$ may differ from $\Lambda_2$, 
     and any vertex $v$ proper to $\sigma_{\Lambda_1}$ and $\tau_{\Lambda_2}$, we have 
    \[
        | P_{H,v}^{\sigma_{\Lambda_1}}(\lambda) - P_{H,v}^{\tau_{\Lambda_2}}(\lambda) | \leq C r^{-d_H(v,\Lambda_1 \neq \Lambda_2)}
    \] 
    where $\Lambda_1 \neq \Lambda_2$ is the set $(\Lambda_1\setminus\Lambda_2)\cup (\Lambda_2\setminus\Lambda_1)\cup \{v\in\Lambda_1\cap\Lambda_2:\sigma_{\Lambda_1}(v)\neq\tau_{\Lambda_2}(v)\}$.
\end{theorem}

\subsubsection{LDC}

In \cite{regts2023absence}, Regts establishes the LDC for the hardcore model (the independence polynomial of a graph) via cluster expansion techniques. However, the result is not immediately clear for general hypergraphs.
Our technique for establishing divisibility also extends to hypergraphs. By constructing a bijection, we prove the following divisibility lemma, which subsequently leads to the LDC.

\begin{lemma}\label{lem:hyper-divide}
    Let $H=(V,E)$ be a hypergraph, $\sigma_\Lambda$ be a partial configuration on $\Lambda \subseteq V$,
    $u,v$ be two distinct vertices in $V \backslash \Lambda$, 
    then 
    \[   
   \lambda^{d_H(u,v)+1} \mid 
   \Bigl(
    Z_{H,u,v}^{\sigma_{\Lambda},+,+}(\lambda)Z_{H,u,v}^{\sigma_{\Lambda}, -,-}(\lambda) 
   -Z_{H,u,v}^{\sigma_{\Lambda}, +,-}(\lambda)Z_{H,u,v}^{\sigma_{\Lambda}, -,+}(\lambda) 
   \Bigr).
    \]
\end{lemma}

\begin{proof}
    Let $\mathcal{I}_1,\mathcal{I}_2,\mathcal{I}_3,\mathcal{I}_4$ denotes the sets of independent sets
    admitting $Z_{H,u,v}^{\sigma_{\Lambda},+,+}(\lambda)$, $Z_{H,u,v}^{\sigma_{\Lambda},-,-}(\lambda)$, $Z_{H,u,v}^{\sigma_{\Lambda},+,-}(\lambda)$ and $Z_{H,u,v}^{\sigma_{\Lambda},-,+}(\lambda)$ respectively.
    Then 
    \begin{align*}
        &  Z_{H,u,v}^{\sigma_{\Lambda},+,+}(\lambda)Z_{H,u,v}^{\sigma_{\Lambda}, -,-}(\lambda) 
        -Z_{H,u,v}^{\sigma_{\Lambda}, +,-}(\lambda)Z_{H,u,v}^{\sigma_{\Lambda}, -,+}(\lambda)  \\ 
        =& \sum_{I_1\in \mathcal{I}_1}\sum_{I_2\in \mathcal{I}_2} \lambda^{|I_1| + |I_2|} - \sum_{I\in \mathcal{I}_3} \sum_{I\in \mathcal{I}_4} \lambda^{|I_3|+|I_4|} \\
    \end{align*}

    Let $A$ be the set of $(I_1,I_2)$ such that $|I_1| + |I_2| < d_H(u,v) + 1$ and $B$ be the set of $(I_3,I_4)$ such that $|I_3| + |I_4| < d_H(u,v) + 1$.
    We will construct a bijection between $A$ and $B$ and if $(I_3,I_4) = f(I_1,I_2)$, then $|I_1| + |I_2| = |I_3| + |I_4|$.
    For any $(I_1,I_2)\in A$, since $|I_1| + |I_2| < d_H(u,v) + 1$, $u$ and $v$ are disconnected in the induced sub-hypergraph $H_{I_1\cup I_2}$. Let $S$ be 
    the connected component in $H_{I_1\cup I_2}$ containing $u$ and $T = V\backslash S$. 
    Then $(I_3,I_4) = (I_1 |_ S \cup I_2|_ T, I_1|_T \cup I_2|_ S) \in B$. Certainly $|I_3| + |I_4| = |I_1| + |I_2|$ and 
    the operation is reversible since $I_3\cup I_4 = I_1\cup I_2$ is unchanged after the swap.
\end{proof}

The divisibility relation directly implies the so-called point-to-point LDC in \cite{shao2024zero}. Moreover,
by induction on the size $|\sigma_{\Lambda_1}\neq \tau_{\Lambda_2}|$, the point-to-point LDC implies the LDC.

\begin{lemma}[Point-to-point LDC]
    Let $H=(V,E)$ be a hypergraph, $\sigma_{\Lambda_1},\tau_{\Lambda_2}$ be two partial configurations on $\Lambda_1,\Lambda_2 \subseteq V$,
    $v$ be a proper vertex to $\sigma_{\Lambda_1}$ and $\tau_{\Lambda_2}$, then 
    \[  
    \lambda^{d_H(v,u)+1} \mid \Bigl(P_{H,v}^{\sigma_{\Lambda_1}}(\lambda) - P_{H,v}^{\sigma_{\Lambda_1},u^+}(\lambda)\Bigr) \quad \text{and} 
    \quad \lambda^{d_H(v,u)+1} \mid \Bigl(P_{H,v}^{\sigma_{\Lambda_1}}(\lambda) - P_{H,v}^{\sigma_{\Lambda_1},u^-}(\lambda)\Bigr).
    \]    
\end{lemma}



\begin{lemma}[\Cref{lem:hyper-LDC-intro} restated]\label{lem:hyper-LDC}
    Let $H=(V,E)$ be a hypergraph, $\sigma_{\Lambda_1},\tau_{\Lambda_2}$ be two partial configurations on $\Lambda_1,\Lambda_2 \subseteq V$,
    $v$ be a proper vertex to $\sigma_{\Lambda_1}$ and $\tau_{\Lambda_2}$, then 
    \[  
    \lambda^{d_H(v,\Lambda_1 \neq \Lambda_2)+1} \mid \Bigl(P_{H,v}^{\sigma_{\Lambda_1}}(\lambda) - P_{H,v}^{\tau_{\Lambda_2}}(\lambda)\Bigr).
    \]
\end{lemma}

\subsubsection{Uniform bound}


Galvin and coauthors~\cite{galvin2024zeroes} proved that the independence polynomial of a hypergraph with maximum degree $\Delta$ is
zero-free in the disk $\D_{\lambda_s(\Delta+1)}$.
Later, Bencs and Buys \cite{bencs2025optimal} improve the zero-free region to $\D_{\lambda_s(\Delta)}$ 
and provide another zero-free region around the Shearer's bound $(0,\lambda_c(\Delta))$, 
extending the result from graphs to hypergraphs as shown in \cite{peters2019conjecture}.
With the zero-free region and the LDC established in \Cref{lem:hyper-LDC}
, we can extend the result in \cite{regts2023absence} of the graph independence polynomial to hypergraphs.

\begin{lemma}[Theorem 1.1 in \cite{bencs2025optimal}]\label{lem:hyper-zero-disk}
    Let $\Delta \geq 2$. For any hypergraph $H=(V,E)$ with maximum degree at most $\Delta$ and 
    $\vlambda \in \C^V$ with $|\lambda_v| \leq \lambda_s(\Delta)$ for all $v\in V$ we have $Z_H(\lambda) \neq 0$.
\end{lemma}

\begin{lemma}[Theorem 1.2 in \cite{bencs2025optimal}]\label{lem:hyper-zero-interval}
    Let $\Delta \geq 3$. There exists an open neighborhood $U_\Delta$ of the interval $(0,\lambda_c(\Delta))$ such that for any hypergraph $H=(V,E)$ with maximum degree 
    at most $\Delta$ and $\lambda \in U$ we have $Z_H(\lambda) \neq 0$.
\end{lemma}

\begin{lemma}\label{lem:hyper-avoid01}
    
    Let $\Delta \geq 3$, $H$ be a hypergraph with maximum degree at most $\Delta$, $\sigma_\Lambda$ be a partial configuration on $\Lambda \subseteq V$, $v$ be a proper vertex to $\sigma_\Lambda$.
    If $\lambda \in (\D_{\lambda_s(\Delta)} \cup U_\Delta) \backslash \{0\}$, then $P_{H,v}^{\sigma_\Lambda}(\lambda)$ avoids $0$ and $1$.
\end{lemma}

\begin{proof}

    Since $P_{H,v}^{\sigma_\Lambda}(\lambda)= P_{H[\sigma_\Lambda],v}(\lambda)$, prove $P_{H,v}(\lambda)$ always avoids $0$ and $1$ is enough.
    By \Cref{lem:hyper-zero-disk,lem:hyper-zero-interval}, $Z_H(\lambda),Z_{H\backslash v}(\lambda)$ and $Z_{H \ominus v}(\lambda) \neq 0$. 
    Thus $P_{H,v}(\lambda) = \frac{\lambda Z_{H\ominus v}(\lambda)}{Z_H(\lambda)} \neq 0$ and $P_{H,v}(\lambda) = 1- \frac{Z_{H \backslash v}(\lambda)}{Z_H(\lambda)} \neq 1$. We are done.
\end{proof}

\begin{lemma}[uniform bound]\label{lem:hyper-bound}
    Fix $\Delta \geq 3$, let $U =(\D_{\lambda_s(\Delta)} \cup U_\Delta)\backslash \{0\}$, and let $S$ be a compact subset of $U$.
    There exists a constant $C>0$ such that for any hypergraph $H=(V,E)$ with maximum degree at most $\Delta$, any $v\in V$, and any $\lambda \in S$, we have
    $ |P_{H,v}(\lambda)| \leq C $.
\end{lemma}

\begin{proof}
    By \Cref{lem:hyper-avoid01}, $P_{H,v}(\lambda)$ always avoids $0$ and $1$ for $\lambda \in U$.
    Pick $\lambda' \in (0,\lambda_s(\Delta))$, then $P_{H,v}(\lambda')$ is a probability and hence contained in $[0,1]$.
    Then by \Cref{lem:montel-bound}, we obtain the upper bound.
\end{proof}

If the zero-free region is not a disk, it seems that we cannot apply \Cref{lem:bound strip}
to deduce that any fixed $\lambda$ in the zero-free region exhibits SSM. However, by the Riemann mapping theorem,
we can transform an arbitrary zero-free region into a unit disk and then apply \Cref{lem:bound strip},
where the LDC and uniform bound still hold. For details, see \cite{regts2023absence,shao2024zero}.
\begin{proof}[Proof of \Cref{thm:hyper-ssm}]
	This follows from the argument of the Riemann mapping theorem and the results of \Cref{lem:hyper-LDC,lem:hyper-bound,lem:bound strip}.
\end{proof}

\subsubsection{FPTAS}

The computation tree of the hypergraph independence polynomial introduced in \cite{liu2014fptas,lu2015fptas} is the key tool to derive FPTAS from the SSM property.
We don't give the exact construction of the computation tree here, but utilize it as a black box.

\begin{theorem}\label{thm:hyper-computation-tree}
    Let $H=(V,E)$ be a hypergraph of maximum degree $\Delta$.
    Then there exists a hypertree $T$ with root $v$ and maximum degree at most $\Delta$ such that $P_{H,v}(\lambda) = P_{T,v}(\lambda)$.
    If size of edges in $H$ is at most $k$, let $T_k$ be the truncation of $T$ at depth $d$ from $v$, 
    then we can compute $P_{T_d,v}(\lambda)$ exactly in time $O(|V| (k\Delta)^d)$.
\end{theorem}

\begin{lemma}
    If $H$ is a hypergraph, $v$ is a vertex in $H$ and $\lambda > 0$ , then $0 \leq P_{H,v}(\lambda) \leq \frac{\lambda}{1+\lambda}$.
\end{lemma}

\begin{proof}
    If $I$ is an independent set in $H$ containing $v$ with weight $w$, then $I\backslash \{v\}$ is still an independent set in $H$ with 
    weight $w/\lambda$. Thus $Z_{H}(\lambda) = Z_{H,v}^+(\lambda) + Z_{H,v}^-(\lambda) \geq  Z_{H,v}^+(\lambda)(1 + 1/\lambda)$,
    which implies $P_{H,v}(\lambda) \leq \frac{\lambda}{1+\lambda}$.
\end{proof}

\begin{lemma}\label{lem:hyper-lower-bound}
    Fix $\Delta \geq 3$ and let $S$ be a compact subset of $(\D_{\lambda_s(\Delta)} \cup U_\Delta) \backslash \{0\}$,
    there exists a constant $C>0$ such that for any hypergraph $H=(V,E)$ with maximum degree at most $\Delta$, 
    any $v\in V$, any $\lambda \in S$, we have $\abs{1-P_{H,v}(\lambda)} \geq C$.
\end{lemma}

\begin{proof}
    Note $P_{H,v}(\lambda)$ always avoids $0$ and $1$ for $\lambda \in (\D_{\lambda_s(\Delta)} \cup U_\Delta) \backslash \{0\}$.
    Then $\frac{1}{1-P_{H,v}(\lambda)}$ is analytic for $\lambda \in (\D_{\lambda_s(\Delta)} \cup U_\Delta) \backslash \{0\}$ and 
     always avoids $0$ and $1$. 
    And pick a positive constant $\lambda' \in (0,\lambda_s(\Delta))$, then $ 0 \leq P_{H,v}(\lambda') \leq \frac{\lambda'}{1+\lambda'}$ always holds,
    i.e. $1 \leq \frac{1}{1-P_{H,v}(\lambda)} \leq 1+\lambda'$. Then by \Cref{lem:montel-bound}, we obtain an upper bound on $\abs{\frac{1}{1-P_{H,v}(\lambda)}}$,
    and hence the lower bound on $\abs{1-P_{H,v}(\lambda)}$.
\end{proof}

\begin{theorem}[FPTAS]\label{thm:hyper-fptas}
    Fix $\Delta \geq 3$, $k \geq 2$ and  $\lambda \in \D_{\lambda_s(\Delta)} \cup U_\Delta$,
    there exists an FPTAS for the hypergraph independence polynomial $Z_H(\lambda)$ 
    for any hypergraph $H=(V,E)$ with maximum degree at most $\Delta$ and maximum edge size at most $k$.
\end{theorem}

\begin{proof}
    
    When $\lambda = 0$, the problem is trivial. Consider $\lambda \neq 0$. 
    Write $V=\{v_1,\ldots,v_n\}$, let $\Lambda_i = \{v_1,\ldots,v_i\}$ and $\sigma_{i}$ be the partial configuration which 
    maps all vertices in $\Lambda_i$ to $-$ (for $i=0,\ldots,n$). Then 
    \begin{align*}
        \frac{1}{Z_H(\lambda)}  &= \frac{Z_H^{\sigma_{n}}(\lambda)}{Z_H^{\sigma_{0}}(\lambda)} 
                                 = \prod_{i=1}^{n} \frac{Z_H^{\sigma_{i}}(\lambda)}{Z_H^{\sigma_{i-1}}(\lambda)} 
                                 = \prod_{i=1}^{n} \left[1 - P_{H,v_i}^{}(\lambda) \right]
                                 = \prod_{i=1}^{n} \left[ 1 - P_{H[\sigma_{i-1}],v_i}(\lambda)  \right].
    \end{align*}

    To approximate $Z_H(\lambda)$ with factor $\eps$, approximating $1 - P_{H[\sigma_{i-1}],v_i}(\lambda)$ with factor $\frac{\eps}{n}$ is enough.
    By \Cref{lem:hyper-lower-bound}, additive  error $\frac{C\eps}{2n}$ for some constant $C>0$ is enough. 
    Then by computation tree in \cite{liu2014fptas} and the SSM result, 
    truncating the computation tree at depth $O(\log \frac{n}{\eps})$ (one way is pinning all vertices at $O(\log \frac{n}{\eps})$ depth to $(-)$) and using the SSM result, the running time is $\mathrm{poly}(\frac{n}{\varepsilon})$, 
    thus we can get the FPTAS.
\end{proof}

\begin{remark}
    In \cite{liu2014fptas,lu2015fptas}, the authors prove a
    computationally efficient correlation decay for $\lambda \in (0,\lambda_c(\Delta))$, 
    which leads to a faster decay rate when dealing with hyperedges of larger sizes. 
    Then they derive an FPTAS for hypergraphs with bounded degree but unbounded edge size
    based on this correlation decay result.
\end{remark}

\subsection{SSM for binary symmetric Holant problems}\label{sec:holant}

Let $G = (V,E)$ be a graph of maximum degree $\Delta$. We consider the Holant problem in the binary symmetric case, which we now describe. Let $\{f_{v}\}_{v \in V}: \N \to \R_{\ge 0}$ be a family of functions, one for each vertex $v \in V$ in the input graph. One should think of each $f_{v}$ as representing a local constraint on the assignments to edges incident to $v$. Since we are restricting ourselves to the binary case, our configurations $\sigma$ will map edges to $\{0,1\}$ (or $-$ and $+$ spins). Furthermore, since we are restricting ourselves to the symmetric case, our local functions $f_{v}$ will only depend on the number of edges incident to $v$ which are mapped to $1$. With these $\{f_{v}\}_{v\in V}$ in hand, we may write the multivariate partition function as
\[
	Z_{G}(\vlambda)=\sum_{\sigma: E \to \{0,1\}} \prod_{v \in V} f_{v}(|\sigma_{E(v)}|) \prod_{e \in E, \sigma(e)=1} \lambda_e,
\]
where $E(v)$ is the set of all edges adjacent to $v$, $\sigma_{E(v)}$ is the configuration restricted on $E(v)$, and $|\sigma_{E(v)}|$ is the number of edges in $E(v)$ with assignment $1$. 

This class of problems is already incredibly rich and encompasses many classical objects studied in combinatorics and statistical physics. As stated in \cite{chen2024spectral}, the weighted even subgraphs model and the Ising model on line graphs are included
for certain choices of $f_v$.
\begin{itemize}
\item \textit{Weighted Even Subgraphs:} In this case, all $f_{v}$ are the same and given by the weighted ``parity'' function. More specifically, for a fixed positive parameter $\rho > 0$, we have
	\[	
	f_{v}(k) = 
		\begin{cases}
		1, & \text{if~$k$ is even};\\
		\rho, & \text{if~$k$ is odd}.
		\end{cases}
	\]
In the case $\rho = 0$, then $Z_{G}(\mathbf{1})$ counts the number of even subgraphs, that is, subsets of edges such that all vertices have even degrees in the resulting subgraph.
\item \textit{Ising Model on Line Graphs:}
In this case, each $f_{v}$ depends on the degree of $v$. If $\beta > 0$ is some fixed parameter (independent of $v$), and $d = \deg(v)$, then we have
\begin{align*}
    f_{v}(k) &= \begin{cases}
        \beta^{\binom{k}{2}}\beta^{\binom{d-k}{2}},
        &\text{if~$0 \leq k \leq d$}; \\
        0, &\text{otherwise.}
    \end{cases}
\end{align*}
\end{itemize}

Let $G=(V,E)$ be a graph, $e\in E$ an edge, and $\sigma_\Lambda$ a partial configuration on $\Lambda \subseteq E\backslash \{e\}$. As in the random cluster model, the conditional probability that $e$ is pinned $+$ is given by  
$P_{G,e}^{\sigma_\Lambda}(\lambda) = {Z_{G,e}^{\sigma_\Lambda,+}(\lambda)}/{Z_{G}^{\sigma_\Lambda}(\lambda)}$. The strong spatial mixing can also be defined as below.

\begin{definition}[SSM for binary symmetric Holant problems]
	Fix the local function $f_v$ and the parameter $\lambda$.
    Let $\G$ be a family of graphs. 
    The Holant problem defined on $\G$ with $f_v$ and $\lambda$
    is said to satisfy strong spatial mixing with exponential rate $r>1$ if
	there exists a constant $C$ such that for any $G=(V,E) \in \G$, any edge $e \in E$, any partial configuration
	$\sigma_{\Lambda_1}$ and $\tau_{\Lambda_2}$ where $\Lambda_1,\Lambda_2 \subseteq E\backslash e$, we have
	\[	
		\abs{P_{G,e}^{\sigma_{\Lambda_1}}(\lambda) - P_{G,e}^{\tau_{\Lambda_2}}(\lambda)} 
		\leq Cr^{-d_G(e,\sigma_{\Lambda_1} \neq \tau_{\Lambda_2})}.
	\]
\end{definition}

\subsubsection{Zerofree}

In \cite{chen2024spectral}, to establish the spectral independence property from zero-freeness, the authors prove the zero-free region for weighted even subgraphs and the Ising model on line graphs. Notably, in \cite{chen2024spectral}, the authors exclude the $\lambda$ factor when pinning an edge to $+$ (or $1$), whereas we do not. Consequently, we exclude the point $0$ from their zero-free region.


\begin{lemma}
    Fix $\rho \in (0,1)$ and $\Delta \in \mathbb{N}^+$, then there exists 
    a complex region $U$ containing $[0,\infty)$ such that for 
    all graphs $G$ with bounded degree $\Delta$ and all partial configurations $\sigma$, the
    partition function of the weighted even subgraph model satisfies
    \(
        Z_G(\lambda) \neq 0
    \)
    for any $\lambda \in U \setminus \{0\}$.
\end{lemma}

\begin{lemma}
	Fix $\beta \in (0,1)$ and $\Delta \in \N^+$, then there exists 
    a complex region $U$ containing $[0,\infty)$ such that for 
    all line graphs $G$ with bounded degree $\Delta$ and all partial configurations $\sigma$, the
    partition function of the antiferromagnetic Ising model satisfies
    \(
        Z_G(\lambda) \neq 0
	\)
    for any $\lambda \in U \setminus \{0\}$.
\end{lemma}

Note that $P_{G,e}^{\sigma_{\Lambda}}$ is analytic in the region $\lambda \in U$. One only needs to check that the ratio is well defined at $\lambda=0$. This holds because $Z_{G,e}^{\sigma_\Lambda,+}$ has a higher order of $\lambda$ than $Z_{G}^{\sigma_\Lambda}$.

\subsubsection{LDC}
\begin{lemma}[LDC for binary symmetric Holant problems]
Let $G=(V,E)$ be a graph, $\sigma$ be a partial configuration on $\Lambda \subseteq E$,
$e_1$ and $e_2$ be two different edges in $E\backslash \Lambda$, then 
\[
	\lambda^{d_G(e_1,e_2)+2} \mid 
    \Bigl(
	Z^{\sigma,+,+}_{G,e_1,e_2}Z^{\sigma,-,-}_{G,e_1,e_2} - Z^{\sigma,+,-}_{G,e_1,e_2}Z^{\sigma,-,+}_{G,e_1,e_2}
    \Bigr).
\]	
\end{lemma}

\begin{proof}
Let $S_1$ be the set of configurations that agree with $Z^{\sigma,+,+}_{G,e_1,e_2}$, and similarly define 
 $S_2,S_3$ and  $S_4$. Then, we have
\begin{align*}
	& Z^{\sigma,+,+}_{G,e_1,e_2}Z^{\sigma,-,-}_{G,e_1,e_2} - Z^{\sigma,+,-}_{G,e_1,e_2}Z^{\sigma,-,+}_{G,e_1,e_2}\\
  = & \sum_{\sigma_1 \in S_1,\sigma_2\in S_2}w(\sigma_1)w(\sigma_2) - \sum_{\sigma_3 \in S_3,\sigma_4\in S_4}w(\sigma_3)w(\sigma_4)
\end{align*}

Define $A = \{(\sigma_1,\sigma_2)\mid n_+(\sigma_1)+n_+(\sigma_2)\leq d_G(e_1,e_2)+1,\sigma_1 \in S_1, \sigma_2\in S_2\}$ and similarly define $B \subseteq S_3\times S_4$. We show there exists a bijection $f:A\to B$ such that if $(\sigma_3,\sigma_4) = f(\sigma_1,\sigma_2)$ then
$w(\sigma_1)w(\sigma_2)=w(\sigma_3)w(\sigma_4)$. If $n_+(\sigma_1)+n_+(\sigma_2)\leq d_G(e_1,e_2)+1$, then
the subgraph $H=(V,\sigma_1 | \sigma_2)$ is disconnected.
Pick $S$ as the connected component containing $e_1$, and let $T=E\backslash S$. 
Define $(\sigma_3,\sigma_4) = (\sigma_1 |_S \cup \sigma_2|_T,\sigma_2 |_S \cup \sigma_1|_T)$, then $f(\sigma_1,\sigma_2)
= (\sigma_3,\sigma_4)$ satisfies our requirements. Firstly $f$ is bijection since $\sigma_1 | \sigma_2 = \sigma_3 | \sigma_4$,
the process is fully reversible. Since there are no $+$ edge between $S$ and $T$ in $\sigma_i (i=1,2,3,4)$,
the local functions $f_v$ for each $v\in S$ are determined by $\sigma_i[S]$ and similarly for $v\in T$. Thus,
\begin{align*}
	w_G(\sigma_1)w_G(\sigma_2) =& w_{G[S]}(\sigma_1 | _S)w_{G[T]}(\sigma_1 | _T)w_{G[S]}(\sigma_2 | _S)w_{G[T]}(\sigma_2 | _T)\\
	=& w_{G[S]}(\sigma_1 | _S)w_{G[T]}(\sigma_2 | _T)w_{G[S]}(\sigma_2 | _S)w_{G[T]}(\sigma_1 | _T) \\
	=& w_{G[S]}(\sigma_3 | _S)w_{G[T]}(\sigma_3 | _T)w_{G[S]}(\sigma_4 | _S)w_{G[T]}(\sigma_4 | _T) \\
	=& w_{G}(\sigma_3)w_G(\sigma_4).
\end{align*}
\end{proof}

%

The divisibility relation implies point-to-point LDC, which then extends to LDC by induction.
The definition of LDC in the Holant framework follows the same description as in the random cluster model.

\subsubsection{SSM}
As before, pick any $\lambda > 0$ at which the ratio is uniformly bounded (as a probability); then we can deduce a uniform
bound on a compact subset by \Cref{lem:montel-bound} from the zero-freeness result. Then, following Regts's approach, we can establish the SSM property for binary symmetric Holant problems once the zero-free region is well understood. 

\begin{theorem}
	Fix $\rho \in (0,1)$, $\Delta \in \mathbb{N}^+$ and $\lambda > 0$. 
	Then the weighted even subgraph model on all graphs $G$ with bounded degree $\Delta$ exhibits SSM.
\end{theorem}

\begin{theorem}

	Fix $\beta \in (0,1)$, $\Delta \in \N^+$ and $\lambda > 0$. 
	Then the Ising model for all line graphs $G$ with bounded degree $\Delta$ exhibits SSM.
\end{theorem}

\subsection{Edge-type SSM for the Potts model}\label{sec:potts-ssm}
The partition function of the Potts model (without external field) of a graph $G=(V,E)$ is defined as
\[
    Z_G(q, \vw ) = \sum_{\sigma:V\ra [q]} \prod_{\substack{(u,v) \in E \\ \sigma(u)=\sigma(v)}} w_{(u,v)}.
\]
where $[q] = \{1,2,\ldots,q\}$, $q$ is the number of colors, $\vw = (w_e)_{e\in E}$ is the edge activity vector.
In the univariate case, write $w = 1+z$, then the partition function of the Potts Model can be written in the form of the Tutte polynomial \cite{sokal2005multivariate} as
\[
    Z_G(q, w ) = \sum_{F\subseteq E} q^{\kappa(V,F)} z^{|F|},
\]
where $\kappa(V,F)$ is the number of connected components of the spanning subgraph $(V,F)$.

Similar to the edge-type SSM for the Ising model in \Cref{cor:edge-ssm},
define the ratio of the partition function of the Potts model as
\[  
    P_{G,e}(w) = \frac {Z_{G-e}(q,w)}{Z_G(q,w)}.
\]
We can prove the Potts model exhibits edge-deletion SSM,
where the constant $\eta \geq 0.002$ is from the zero-free region \cite{bencs2024deterministic}.

\begin{theorem}[Edge-type SSM for the Potts model]\label{thm:potts-ssm}
    Fix $\Delta \in \N$, $q \geq  (2-\eta)(2\Delta-2)$ and $w \in [0,1]$,
    then there exist constants $C>0$ and $r>1$ such that
    for any graph $G=(V,E)$ with maximum degree at most $\Delta$, $e\in E$, $A,B \subseteq E\bs \{e\}$, we have
    \[
        \left|P_{G-A,e}(q,w) - P_{G-B,e}(q,w)\right|\leq Cr^{-d_G(e,A\neq B)}.
    \]
\end{theorem}

In \cite{bencs2024deterministic}, the zero-free region of the univariate  Potts model is studied, and the authors claimed that it also works in the multivariate setting.

\begin{lemma}[Theorem~1 and Section~8 in \cite{bencs2024deterministic}]\label{lem:zerofree-multi-potts}
There exists a constant $\eta \geq 0.002$ such that for all integers $\Delta \geq 3$ and 
$q \geq (2 - \eta)\Delta$ there exists an open set $U_{\Delta,q} \subseteq \mathbb{C}$ containing 
the interval $[0,1]$ such that for any graph $G=(V,E)$ of maximum degree at most $\Delta$ and $\vw \in (U_{\Delta,q})^E$ and
we have $Z_G(q,\vw) \neq 0$.
\end{lemma}


By \Cref{lem:zerofree-multi-potts}, we can get the following result. For $A,B \subseteq \C$, define $A \cdot B = \{ab \mid a\in A, b\in B\}$.
We can immediately obtain that there exists an open set $\mathcal{U}_{\Delta,q} \subseteq U_{\Delta,q}$ containing the real \emph{closed} interval $[0,1]$ and
$\mathcal{U}_{\Delta,q} \cdot  \mathcal{U}_{\Delta,q} \subseteq U_{\Delta,q}$. The open set $\mathcal{U}_{\Delta,q}\backslash \{1\}$ guarantees that the ratio $P_{G,e}(w) = \frac {Z_{G-e}(q,w)}{Z_G(q,w)}$ avoids $0$ and $1$.

\subsubsection{Uniform bound}

\begin{lemma}
    If $\Delta \in \N$, $q \geq  (2-\eta)(2\Delta-2)$, $w \in \mathcal{U}_{2\Delta-2,q} \backslash \{1\}$,
    $G=(V,E)$ is a graph with maximum degree at most $\Delta$ and $e \in E$, then
    $P_{G,e}(w)$ avoids $0$ and $1$.
\end{lemma}

\begin{proof}
    By \Cref{lem:zerofree-multi-potts}, $P_{G,e}(w)\neq 0$ is trivial. We prove $P_{G,e}(w)\neq 1$.
    Let $e = (u,v)$, then
    \begin{align*}
            & Z_G(q,w) - Z_{G-e}(q,w) \\
        =   & \sum_{\sigma \in [q]^V} \prod_{\substack{(x,y)\in E \\ \sigma(x) = \sigma(y)}} w - 
              \sum_{\sigma \in [q]^V} \prod_{\substack{(x,y)\in E-e \\ \sigma(x) = \sigma(y)}} w \\
        =   &  (w-1)  \sum_{ \substack{\sigma \in [q]^V \\ \sigma(u) = \sigma(v)}}\prod_{\substack{(x,y)\in E-e \\ \sigma(x) = \sigma(y)}} w. 
    \end{align*}  
    Thus we can construct $G'=(V',E')$ from $G$ by 
    merging $u,v$ into a single vertex $x$, if parallel edges $(u,y)$ and $(v,y)$ exist in $G$, merge them into a single edge and set $w_{(x,y)} = w_{(u,y)}w_{(v,y)}$.
    Then $Z_G(q,w) - Z_{G-e}(q,w) = (w-1)Z_{G'}(q,\vw)$, where $\vw$ is the edge activity vector of $G'$. 
    Note $\vw \in (U_{2\Delta-2,q})^{E'}$ and $G'$ has maximum degree at most $2\Delta-2$, since $q \geq  (2-\eta)(2\Delta-2)$, 
     by \Cref{lem:zerofree-multi-potts}, $Z_{G'}(q,\vw) \neq 0$ and hence $P_{G,e}(w) \neq 1$.
\end{proof}

Pick a small enough $\eps > 0$ such that $1+\eps \in \mathcal{U}_{2\Delta-2,q}$, one can see that $0 < P_{G,e}(1+\eps) < 1$ always holds. 
Then by \Cref{lem:montel-bound}, we obtain a uniform bound on the ratio of partition functions in the Potts model.

\begin{lemma}\label{lem:potts_uniform_bound}
    Fix $\Delta \in \N$, and let $S$ be a compact subset of $\mathcal{U}_{2\Delta-2,q} \backslash \{1\}$. There exists a constant $C>0$ such that 
    for any graph $G=(V,E)$ with maximum degree at most $\Delta$, any $q \geq  (2-\eta)(2\Delta-2)$, any $e\in E$, any $w\in S$, we have
    $|P_{G,e}(w)| \leq C$.
\end{lemma}

\subsubsection{LDC}

\begin{lemma}[LDC for the Potts model]\label{lem:potts-divide}
    Let $G=(V,E)$ be a graph, $e_1,e_2$ be two different edges in $G$, 
    then 
    \[  
    (w-1)^{d_G(e_1,e_2)}  \mid 
    \Bigl(
    Z_{G-e_1}(q,w)Z_{G-e_2}(q,w) -  Z_{G}(q,w)Z_{G-\{e_1,e_2\}}(q,w)
    \Bigr).
    \]
\end{lemma}

\begin{proof}
    Let $z=w-1$, then 
    \begin{align*}
          & Z_{G-e_1}(q,w)Z_{G-e_2}(q,w) -  Z_{G}(q,w)Z_{G-\{e_1,e_2\}}(q,w) \\
        = & \sum_{ \substack{F_1 \subseteq E-e_1,\\ F_2 \subseteq E-e_2}} q^{\kappa(V,F_1)+\kappa(V,F_2)} z^{|F_1|+|F_2|} 
            - \sum_{ \substack{F_3 \subseteq E,\\ F_4 \subseteq E-\{e_1,e_2\}}} q^{\kappa(V,F_3)+\kappa(V,F_4)} z^{|F_3|+|F_4|} 
    \end{align*}

    Let $A$ be the set of $(F_1,F_2)$ in the first sum such that $|F_1| + |F_2| < d_G(e_1,e_2)$ and $B$ be the set of $(F_3,F_4)$ in the second sum such that $|F_3| + |F_4| < d_G(e_1,e_2)$.
    We will show that there exists a bijection $f$ between $A$ and $B$ such that if $(F_3,F_4) = f(F_1,F_2)$, then $|F_3| + |F_4| = |F_1| + |F_2|$ 
    and $\kappa(V,F_3)+\kappa(V,F_4) = \kappa(V,F_1)+\kappa(V,F_2)$.
    
    Let $F_1,F_2$ be a pair in $A$, since $|F_1| + |F_2| < d_G(e_1,e_2)$, then $e_1,e_2$ are disconnected in the subgraph $H=(V,F_1\cup F_2 \cup \{e_1,e_2\})$.
    Consider the connected component $S$ of $H$, which contains $e_1$, and let $T = V\backslash S$. 
    Then $F_3 = F_1 |_ T \cup F_2 |_ S$ and $F_4 = F_1 |_ S \cup F_2 |_ T$ are in $B$. One can check that $(F_3,F_4)$ is the desired pair and 
    the process is reversible (since $F_1\cup F_2 = F_3 \cup F_4$). We are done.
\end{proof}


\begin{lemma}\label{lem:pottsLDC1}
    Let $G=(V,E)$ be a graph, $e\in E$, and $A\subseteq E\bs \{e\}$, then the Taylor series near $w=1$ of  $P_{G,e}(q,w)$ and $P_{G-A,e}(q,w)$ satisfies
    \[  
        (w-1)^{d_G(e,A)}  \mid  \Bigl(P_{G,e}(q,w) - P_{G-A,e}(q,w)\Bigr).
    \] 
\end{lemma}

\begin{proof}
    We prove this by induction on $|A|$. The base case $|A|=1$, for instance $A=\{e'\}$,
    \begin{align*}
         P_{G,e}(q,w) - P_{G-e',e}(q,w) 
        =& \frac{Z_{G-e}(q,w)}{Z_G(q,w)} - \frac{Z_{G-\{e,e'\}}(q,w)}{Z_{G-e'}(q,w)} \\
        =& \frac{Z_{G-e}(q,w)Z_{G-e'}(q,w) - Z_{G}(q,w)Z_{G-\{e,e'\}}(q,w)}{Z_G(q,w)Z_{G-e'}(q,w)}.
    \end{align*}

    Clearly $\frac{1}{Z_G(q,w)Z_{G-e'}(q,w)}$ is analytic near $w=1$. By \Cref{lem:potts-divide}, we have $(w-1)^{d_G(e,e')} \mid \Bigl( P_{G,e}(q,w) - P_{G-e',e}(q,w) \Bigr)$.

    Now consider the case $k\geq 2$, suppose the statement holds for $|A|\leq k-1$, we prove it for $|A|=k$. Pick $e'\in A$, let $A' = A\backslash \{e'\}$, then 
    \[
        P_{G,e}(q,w) - P_{G-A,e}(q,w) = [P_{G,e}(q,w) - P_{G-A',e}(q,w)] + [P_{G-A',e}(q,w) - P_{G-A,e}(q,w)].  
    \]

    By induction hypothesis, we have $(w-1)^{d_G(e,A')} \mid \Bigl(P_{G,e}(q,w) - P_{G-A',e}(q,w)\Bigr)$, and $(w-1)^{d_{G-A'}(e,e')} \mid \Bigl(P_{G-A',e}(q,w) - P_{G-A,e}(q,w)\Bigr)$.
    Since $d_G(e,A) \leq d_G(e,A')$ and $d_G(e,A) \leq d_G(e,e') \leq d_{G-A'}(e,e')$, we have $(w-1)^{d_G(e,A)} \mid \Bigl(P_{G,e}(q,w) - P_{G-A,e}(q,w)\Bigr)$.
\end{proof}

\begin{lemma}\label{lem:pottsLDC2}
    Let $G=(V,E)$ be a graph, $e\in E$, and $A,B \subseteq E\bs \{e\}$, then the Taylor series near $w=1$ of $P_{G-A,e}(q,w)$ and $P_{G-B,e}(q,w)$ satisfies
    \[  
        (w-1)^{d_G(e_1,A\neq B)}  \mid  \Bigl(P_{G-A,e}(q,w) - P_{G-B,e}(q,w)\Bigr).
    \]
\end{lemma}

\begin{proof}
    Let $G'= G-(A\cup B)$, $A'= A\bs B$ and $B'= B\bs A$, then
    \begin{align*}
        P_{G-A,e}(q,w) - P_{G-B,e}(q,w) =& P_{G'-A',e}(q,w) - P_{G'-B',e}(q,w) \\ 
                                        =& [P_{G'-A',e}(q,w) - P_{G',e}(q,w)] + [P_{G',e}(q,w) - P_{G'-B',e}(q,w)].  
    \end{align*}

    By the previous lemma, we have 
    $(w-1)^{d_{G'}(e,A')} \mid \Bigl(P_{G'-A',e}(q,w) - P_{G',e}(q,w)\Bigr)$ and 
    $(w-1)^{d_{G'}(e,B')} \mid \Bigl(P_{G',e}(q,w) - P_{G'-B',e}(q,w)\Bigr)$.
    Since $d_G(e,A\neq B) = \min \{d_{G}(e,A'),d_{G}(e,B')\} \leq \min \{d_{G'}(e,A'),d_{G'}(e,B')\}$, we are done.
\end{proof}

Combining \Cref{lem:pottsLDC2,lem:potts_uniform_bound,lem:bound strip}, 
we can establish the edge-type SSM result for the Potts model (\Cref{thm:potts-ssm}).

\begin{remark}
    Ratio $1/P_{G,e}(w)$ also exhibits edge-type SSM. Since $1/P_{G,e}(w)$ still avoids $0$ and $1$ and the $0 < 1/P_{G,e}(0) = \frac{Z_G(q,0)}{Z_G(q,0) + Z_{G'}(q,0)+Z} < 1$,
    the uniform bound can be obtained by \Cref{lem:montel-bound}. Also, LDC can still be established by the same technique.
\end{remark}

\subsection{Vertex-type SSM for the Ising model}
Similar to the edge-type SSM of the Ising model in \Cref{thm:eSSM}, we also establish a vertex-type SSM.

\begin{theorem}[Vertex-type SSM for Ising model]\label{thm:vSSM}
    Fix edge activity $\beta \geq 1$ and 
    uniform external $\lambda \in \D_\frac{1}{\beta}$ for Ising model, 
    and $c \in [0,1)$.
    Then there exist constants $C>0$ and $r>1$ such that
    for all graphs $G=(V,E)$, $v\in V$, $A,B \subseteq V \backslash \{v\}$,
    let $m = \{\lambda_v \ra c\lambda\}$, 
    $m_1 = \{\lambda_u \ra c\lambda\}_{u\in A}$,
    $m_2 = \{\lambda_u \ra c\lambda\}_{u\in B}$, we have

    \[
        \left|P_{G,m}^{m_1}-P_{G,m}^{m_2}\right|\leq Cr^{-d_G(v,m_1 \neq m_2)}.
    \]
\end{theorem}

\subsubsection{Vertex-type LDC}\label{sec:ising-vertex-ssm}

\begin{lemma}[Vertex-type LDC for the Ising model]

    For $\vbeta \geq 1$, $c \in [0,1)$, let $G=(V,E)$ be a graph, $v\in V$, $A\subseteq V\backslash\{v\}$,
    $\vlambda \in \D^V$,
    $m = \{ \lambda_v z \ra  c \lambda_v z\}$,
    $m_1 = \{ \lambda_u z \ra c \lambda_u z\}_{u\in A}$, 
    then
    \[
        z^{d_G(v,A)+1} \mid \Bigl(P_{G,m}(\beta,\vlambda z) - P_{G,m}^{m_1}(\beta,\vlambda z)\Bigr).
    \]
\end{lemma}

\begin{proof}
    \begin{align*}
        P_{G,m}(\beta,\vlambda z) - P_{G,m}^{m_1}(\beta,\vlambda z) 
        =& \frac{Z_G^m(\beta,\vlambda z)}{Z_G(\beta,\vlambda z)} - \frac{Z_G^{m,m_1}(\beta,\vlambda z)}{Z_G^{m_1}(\beta,\vlambda z)}   \\
        =& \frac{Z_G^m(\beta,\vlambda z)Z_G^{m_1}(\beta,\vlambda z) - Z_G^{m,m_1}(\beta,\vlambda z)Z_G(\beta,\vlambda z)}{Z_G(\beta,\vlambda z)Z_G^{m_1}(\beta,\vlambda z)}.
    \end{align*}

    Clearly $\frac{1}{Z_G(\beta,\vlambda z)Z_G^{m_1}(\beta,\vlambda z)}$ is analytic near $z = 0$. Proof of \Cref{lem:LDC} also 
    apply to vertex, then we have 
    $z^{d_G(v,A)+1} \mid \Bigl(P_{G,m}(\beta,\vlambda z) - P_{G,m}^{m_1}(\beta,\vlambda z)\Bigr)$.
\end{proof}

\begin{lemma}\label{lem:vertexLDC}
    For $\beta > 1$, $c \in [0,1)$, $G=(V,E)$ be a graph, $\vlambda \in \D^V$,
	$v\in V$, $A,B \subseteq V\backslash\{v\}$, 
    $m = \{ \lambda_v z \ra c\lambda_v z\}$,
    $m_1 = \{ \lambda_u z \ra c\lambda_u z \}_{u\in A}$, $m_2 = \{ \lambda_u z \ra c\lambda_uz \}_{u\in B}$, then
    \[
        z^{d_G(v,m_1\neq m_2)+1} \mid \Bigl( P_{G,m}^{m_1}(\beta,\vlambda z) - P_{G,m}^{m_2}(\beta,\vlambda z)\Bigr)
    \]
    where $m_1 \neq m_2$ is vertex set where $m_1$ and $m_2$ differ.
\end{lemma}

\begin{proof}
    Consider $\vlambda' z$ as the uniform external field $\vlambda z$ applied  $m_1 \cap m_2$, let $m_1' = m_1 \backslash m_2$, $m_2' = m_2 \backslash m_1$, then
    \begin{align*}
        P_{G,m}^{m_1}(\beta,\vlambda z) - P_{G,m}^{m_2}(\beta,\vlambda z) 
        =& P_{G,m}^{m_1'}(\beta,\vlambda' z) - P_{G,m}^{m_2'}(\beta,\vlambda' z) \\
        =& [ P_{G,m}^{m_1'}(\beta,\vlambda' z) - P_{G,m}(\beta,\vlambda'z) ] + [ P_{G,m}(\beta,\vlambda'z) - P_{G,m}^{m_2'}(\beta,\vlambda'z) ].
    \end{align*}

    By the previous lemma, 
    we have $z^{d_G(v,m_1')+1} \mid \Bigl( P_{G,m_1'}(\beta,\vlambda'z
    ) - P_{G,m}(\beta,\vlambda'z)\Bigr)$ and 
    $z^{d_G(v,m_2')+1} \mid \Bigl( P_{G,m}(\beta,\vlambda'z) - P_{G,m_2'}(\beta,\vlambda'z) \Bigr)$. 
    Since $d_G(v, m_1 \neq m_2) = \min \{ d_G(v,m_1'), d_G(v,m_2') \}$, we are done.
\end{proof}

\subsubsection{Uniform bound of vertex-type ratio}

 \begin{lemma}[{\cite[Corollary 40 ]{shao2024zero}}]\label{lem:pin}
     Let $G$ be a graph and $v$ be a vertex in $G$. Then the partition function of the Ising model $Z^+_{G, v}(\beta, \vlambda)$ can be expressed as:
     \[
         Z^+_{G, v}(\beta, \vlambda)=\lambda_v  Z_{G\backslash\{v\}}(\beta, \vlambda^{v^+})
     \]
     where $Z_{G\backslash\{v\}}(\beta, \vlambda^{v^+})$ is the partition function of the Ising model with non-uniform external fields $\vlambda^{v^+}$ on the graph $G\backslash\{v\}$ obtained from $G$ by deleting $v$, and $\lambda^{v^+}_w=\lambda_w$ for $w\in V\backslash (N(v)\cup\{v\})$ and $\lambda^{v^+}_w=\beta\lambda_w$ for $w\in N(v)$.
 \end{lemma}

\begin{lemma}\label{lem:vertexavoid01}
    Let $G=(V,E)$ be a graph, $\beta > 1$ , $\vlambda \in \D_{\frac{1}{\beta}}^{V}$, 
    $v \in V(G)$, if $\lambda' \in \D_{\frac{1}{\beta}}$ and $\lambda'_v \neq \lambda_v$, 
    then $P_{G, \{\lambda_v \ra \lambda'\}}(\beta,\vlambda)$ avoids $0$ and $1$.
\end{lemma}

\begin{proof}

    By the Lee--Yang theorem, it is trivial that $P_{G,\{\lambda_v \ra \lambda'\}}(\beta,\vlambda) \neq 0$.
    We prove the ratio avoids $1$.

    \begin{align*}
          & Z_G(\beta, \vlambda) - Z_G(\beta, \vlambda')                                           \\
        = & Z_{G,v}^{+}(\beta,\vlambda) + Z_{G,v}^{-}(\beta,\vlambda)
        -  Z_{G,v}^{+}(\beta,\vlambda') - Z_{G,v}^{-}(\beta,\vlambda') \\
        = & Z_{G,v}^{+}(\beta,\vlambda) -  Z_{G,v}^{+}(\beta,\vlambda')         \\
        = & (\lambda_v - \lambda'_v) Z_{G\backslash\{v\}}(\beta,\vlambda^{v^+})  \quad  \text{(see \Cref{lem:pin})} 
    \end{align*}

    Since $\vlambda \in \D_\frac{1}{\beta}^{V}$, we have $\vlambda^{v^+} \in \D^{V\backslash \{v\}}$, and by the Lee--Yang theorem, $Z_{G\backslash\{v\}}(\beta,\vlambda^{v^+}) \neq0$, thus the ratio avoids $1$.

\end{proof}

\begin{lemma}\label{lem:vertexbound}
    Fix $\beta \geq 1$ and $c\in [0,1)$,
    then for any compact set $S \subseteq \D_\frac{1}{\beta}\backslash\{0\}$,
    there exists a constant $C$ such that for any graph $G=(V,E)$, vertex $v\in V$, $A \subseteq V\backslash\{v\}$,
    $m = \{ \lambda_v \ra c \lambda_v \}$, $m_1 = \{ \lambda_u \ra c \lambda_u \}_{u\in A}$,
    such that $|P_{G,m}^{m_1}(\beta,\lambda)| \leq C$ for all $\lambda \in S$.
\end{lemma}

\begin{proof}
    By \Cref{lem:vertexavoid01}, $P_{G,m}^{m_1}(\beta,\lambda)$ avoids $0$ and $1$ for all $\lambda \in \D_\frac{1}{\beta}\backslash\{0\}$. Pick a positive constant $\lambda' \in (0,\frac{1}{\beta})$,
    then $ 0 < P_{G,m}^{m_1}(\beta,\lambda') < 1$ always holds.
    Then by \Cref{lem:montel-bound}, we obtain the upper bound.
\end{proof}

Combining \Cref{lem:vertexbound,lem:vertexLDC,lem:bound strip},
we can establish the vertex-type SSM.








\bibliographystyle{alpha}
\bibliography{main}

@article{BAG07,
  title = {Ising model in half-space: A series of phase transitions in low magnetic fields},
  author = {Basuev, A.G.},
  journal = {Theor. Math. Phys.},
  volume = {153},
  pages = {1539--1574},
  year = {2007},
  publisher = {Springer},
}

@inproceedings{Bencszero,
  title = {On complex roots of the independence polynomial},
  author = {Bencs, Ferenc and Csikv{\'a}ri, P{\'e}ter and Srivastava, Piyush and Vondr{\'a}k, Jan},
  booktitle = {Proceedings of the 2023 Annual ACM-SIAM Symposium on Discrete Algorithms (SODA)},
  pages = {675--699},
  year = {2023},
  organization = {SIAM},
  timestamp = {Mon, 26 Jun 2023 01:00:00 +0200},
  biburl = {https://dblp.org/rec/conf/soda/BencsC0V23.bib},
  bibsource = {dblp computer science bibliography, https://dblp.org},
  doi = {10.1137/1.9781611977554.ch29},
  _bib2doi_selected = {dblp:/rec/conf/soda/BencsC0V23.bib},
  _bib2doi_confirmed = {true},
}

@inproceedings{zerofe,
  title = {{Zeros of ferromagnetic 2-spin systems}},
  author = {Guo, Heng and Liu, Jingcheng and Lu, Pinyan},
  booktitle = {ACM-SIAM Symposium on Discrete Algorithms},
  pages = {181--192},
  year = {2020},
  organization = {SIAM},
  timestamp = {Thu, 15 Jul 2021 01:00:00 +0200},
  biburl = {https://dblp.org/rec/conf/soda/00010L20.bib},
  bibsource = {dblp computer science bibliography, https://dblp.org},
  doi = {10.1137/1.9781611975994.11},
  _bib2doi_selected = {dblp:/rec/conf/soda/00010L20.bib},
  _bib2doi_confirmed = {true},
}

@article{GGHP22,
  title = {The complexity of approximating the complex-valued Ising model on bounded degree graphs},
  author = {Galanis, Andreas and Goldberg, Leslie A and Herrera-Poyatos, Andr{\'e}s},
  journal = {SIAM Journal on Discrete Mathematics},
  volume = {36},
  number = {3},
  pages = {2159--2204},
  year = {2022},
  publisher = {SIAM},
  timestamp = {Mon, 03 Mar 2025 00:00:00 +0100},
  biburl = {https://dblp.org/rec/journals/siamdm/GalanisGH22.bib},
  bibsource = {dblp computer science bibliography, https://dblp.org},
  doi = {10.1137/21m1454043},
  _bib2doi_selected = {dblp:/rec/journals/siamdm/GalanisGH22.bib},
  _bib2doi_confirmed = {true},
}

@article{jerrum1993polynomial,
  title = {{Polynomial-time approximation algorithms for the Ising model}},
  author = {Jerrum, Mark and Sinclair, Alistair},
  journal = {SIAM Journal on computing},
  volume = {22},
  number = {5},
  pages = {1087--1116},
  year = {1993},
  publisher = {SIAM},
  timestamp = {Sun, 02 Jun 2019 01:00:00 +0200},
  biburl = {https://dblp.org/rec/journals/siamcomp/JerrumS93.bib},
  bibsource = {dblp computer science bibliography, https://dblp.org},
  doi = {10.1137/0222066},
  _bib2doi_selected = {dblp:/rec/journals/siamcomp/JerrumS93.bib},
  _bib2doi_confirmed = {true},
}

@article{LeeYang1,
  title = {{Statistical theory of equations of state and phase transitions. I. Theory of condensation}},
  author = {Yang, Chen-Ning and Lee, Tsung-Dao},
  journal = {Physical Review},
  volume = {87},
  number = {3},
  pages = {404},
  year = {1952},
  publisher = {APS},
}

@article{LeeYang2,
  title = {{Statistical theory of equations of state and phase transitions. II. Lattice gas and Ising model}},
  author = {Lee, Tsung-Dao and Yang, Chen-Ning},
  journal = {Physical Review},
  volume = {87},
  number = {3},
  pages = {410},
  year = {1952},
  publisher = {APS},
}

@inproceedings{li2012approximate,
  title = {{Approximate counting via correlation decay in spin systems}},
  author = {Li, Liang and Lu, Pinyan and Yin, Yitong},
  booktitle = {Proceedings of the twenty-third annual ACM-SIAM symposium on Discrete Algorithms},
  pages = {922--940},
  year = {2012},
  organization = {SIAM},
  timestamp = {Sat, 30 Sep 2023 01:00:00 +0200},
  biburl = {https://dblp.org/rec/conf/soda/LiLY12.bib},
  bibsource = {dblp computer science bibliography, https://dblp.org},
  doi = {10.1137/1.9781611973099.74},
  _bib2doi_selected = {dblp:/rec/conf/soda/LiLY12.bib},
  _bib2doi_confirmed = {true},
}

@inproceedings{li2013correlation,
  title = {{Correlation decay up to uniqueness in spin systems}},
  author = {Li, Liang and Lu, Pinyan and Yin, Yitong},
  booktitle = {Proceedings of the twenty-fourth annual ACM-SIAM symposium on Discrete algorithms},
  pages = {67--84},
  year = {2013},
  organization = {SIAM},
  timestamp = {Sat, 30 Sep 2023 01:00:00 +0200},
  biburl = {https://dblp.org/rec/conf/soda/LiLY13.bib},
  bibsource = {dblp computer science bibliography, https://dblp.org},
  doi = {10.1137/1.9781611973105.5},
  _bib2doi_selected = {dblp:/rec/conf/soda/LiLY13.bib},
  _bib2doi_confirmed = {true},
}

@article{LSSFisherzeros,
  title = {{Fisher zeros and correlation decay in the Ising model}},
  author = {Liu, Jingcheng and Sinclair, Alistair and Srivastava, Piyush},
  journal = {Journal of Mathematical Physics},
  volume = {60},
  number = {10},
  year = {2019},
  publisher = {AIP Publishing},
}

@article{PRS23,
  title = {A near-optimal zero-free disk for the Ising model},
  author = {Viresh Patel and Guus Regts and Ayla Stam},
  journal = {ArXiv},
  year = {2023},
  volume = {abs/2311.05574},
  timestamp = {Thu, 16 Nov 2023 00:00:00 +0100},
  biburl = {https://dblp.org/rec/journals/corr/abs-2311-05574.bib},
  bibsource = {dblp computer science bibliography, https://dblp.org},
  doi = {10.48550/arXiv.2311.05574},
  _bib2doi_selected = {dblp:/rec/journals/corr/abs-2311-05574.bib},
  _bib2doi_confirmed = {true},
}

@article{petersregts20,
  title = {{Location of zeros for the partition function of the Ising model on bounded degree graphs}},
  author = {Peters, Han and Regts, Guus},
  journal = {Journal of the London Mathematical Society},
  volume = {101},
  number = {2},
  pages = {765--785},
  year = {2020},
  publisher = {Wiley Online Library},
}

@article{sinclair2014approximation,
  title = {{Approximation algorithms for two-state anti-ferromagnetic spin systems on bounded degree graphs}},
  author = {Sinclair, Alistair and Srivastava, Piyush and Thurley, Marc},
  journal = {Journal of Statistical Physics},
  volume = {155},
  number = {4},
  pages = {666--686},
  year = {2014},
  publisher = {Springer},
}

@inproceedings{Weitz06,
  title = {{Counting independent sets up to the tree threshold}},
  author = {Weitz, Dror},
  booktitle = {Proceedings of the thirty-eighth annual ACM symposium on Theory of computing},
  pages = {140--149},
  year = {2006},
  timestamp = {Tue, 06 Nov 2018 00:00:00 +0100},
  biburl = {https://dblp.org/rec/conf/stoc/Weitz06.bib},
  bibsource = {dblp computer science bibliography, https://dblp.org},
  doi = {10.1145/1132516.1132538},
  _bib2doi_selected = {dblp:/rec/conf/stoc/Weitz06.bib},
  _bib2doi_confirmed = {true},
}

@article{zhang2009approximating,
  title = {{Approximating partition functions of the two-state spin system}},
  author = {Zhang, Jinshan and Liang, Heng and Bai, Fengshan},
  journal = {Information Processing Letters},
  volume = {111},
  number = {14},
  pages = {702--710},
  year = {2011},
  publisher = {Elsevier},
  timestamp = {Tue, 21 Mar 2023 00:00:00 +0100},
  biburl = {https://dblp.org/rec/journals/ipl/ZhangLB11.bib},
  bibsource = {dblp computer science bibliography, https://dblp.org},
  doi = {10.1016/j.ipl.2011.04.012},
  _bib2doi_selected = {dblp:/rec/journals/ipl/ZhangLB11.bib},
  _bib2doi_confirmed = {true},
}

@article{mann2019approximation,
  title = {Approximation algorithms for complex-valued Ising models on bounded degree graphs},
  author = {Mann, Ryan L and Bremner, Michael J},
  journal = {Quantum},
  volume = {3},
  pages = {162},
  year = {2019},
  publisher = {Verein zur F{\"o}rderung des Open Access Publizierens in den Quantenwissenschaften},
  timestamp = {Tue, 01 Apr 2025 01:00:00 +0200},
  biburl = {https://dblp.org/rec/journals/quantum/MannB19.bib},
  bibsource = {dblp computer science bibliography, https://dblp.org},
  doi = {10.22331/q-2019-07-11-162},
  _bib2doi_selected = {dblp:/rec/journals/quantum/MannB19.bib},
  _bib2doi_confirmed = {true},
}

@book{barvinok2016combinatorics,
  title = {Combinatorics and complexity of partition functions},
  author = {Barvinok, Alexander},
  volume = {30},
  year = {2016},
  publisher = {Springer},
  timestamp = {Tue, 26 Sep 2017 01:00:00 +0200},
  biburl = {https://dblp.org/rec/books/daglib/0041359.bib},
  bibsource = {dblp computer science bibliography, https://dblp.org},
  doi = {10.1007/978-3-319-51829-9},
  _bib2doi_selected = {dblp:/rec/books/daglib/0041359.bib},
  _bib2doi_confirmed = {true},
}

@article{patel2017deterministic,
  title = {Deterministic polynomial-time approximation algorithms for partition functions and graph polynomials},
  author = {Patel, Viresh and Regts, Guus},
  journal = {SIAM Journal on Computing},
  volume = {46},
  number = {6},
  pages = {1893--1919},
  year = {2017},
  publisher = {SIAM},
  timestamp = {Tue, 13 Mar 2018 00:00:00 +0100},
  biburl = {https://dblp.org/rec/journals/siamcomp/PatelR17.bib},
  bibsource = {dblp computer science bibliography, https://dblp.org},
  doi = {10.1137/16M1101003},
  _bib2doi_selected = {dblp:/rec/journals/siamcomp/PatelR17.bib},
  _bib2doi_confirmed = {true},
}

@article{peters2019conjecture,
  title = {On a conjecture of Sokal concerning roots of the independence polynomial},
  author = {Peters, Han and Regts, Guus},
  journal = {Michigan Mathematical Journal},
  volume = {68},
  number = {1},
  pages = {33--55},
  year = {2019},
  publisher = {University of Michigan, Department of Mathematics},
}

@article{regts2023absence,
  title = {Absence of zeros implies strong spatial mixing},
  author = {Regts, Guus},
  journal = {Probability Theory and Related Fields},
  volume = {186},
  number = {1},
  pages = {621--641},
  year = {2023},
  publisher = {Springer},
}

@article{shao2024zero,
  title = {From Zero-Freeness to Strong Spatial Mixing via a Christoffel-Darboux Type Identity},
  author = {Shao, Shuai and Ye, Xiaowei},
  journal = {arXiv preprint arXiv:2401.09317},
  year = {2024},
}

@article{liu2019ising,
  title = {The Ising partition function: Zeros and deterministic approximation},
  author = {Liu, Jingcheng and Sinclair, Alistair and Srivastava, Piyush},
  journal = {Journal of Statistical Physics},
  volume = {174},
  number = {2},
  pages = {287--315},
  year = {2019},
  publisher = {Springer},
}

@article{shao2021contraction,
  title = {Contraction: A unified perspective of correlation decay and zero-freeness of 2-spin systems},
  author = {Shao, Shuai and Sun, Yuxin},
  journal = {Journal of Statistical Physics},
  volume = {185},
  pages = {1--25},
  year = {2021},
  publisher = {Springer},
}

@article{galvin2024zeroes,
  title = {On the zeroes of hypergraph independence polynomials},
  author = {Galvin, David and McKinley, Gwen and Perkins, Will and Sarantis, Michail and Tetali, Prasad},
  journal = {Combinatorics, Probability and Computing},
  volume = {33},
  number = {1},
  pages = {65--84},
  year = {2024},
  publisher = {Cambridge University Press},
  timestamp = {Thu, 01 May 2025 01:00:00 +0200},
  biburl = {https://dblp.org/rec/journals/cpc/GalvinMPST24.bib},
  bibsource = {dblp computer science bibliography, https://dblp.org},
  doi = {10.1017/s0963548323000330},
  _bib2doi_selected = {dblp:/rec/journals/cpc/GalvinMPST24.bib},
  _bib2doi_confirmed = {true},
}

@article{bencs2024deterministic,
  title = {Deterministic approximate counting of colorings with fewer than $2 \backslash Delta $ colors via absence of zeros},
  author = {Bencs, Ferenc and Berrekkal, Khallil and Regts, Guus},
  journal = {arXiv preprint arXiv:2408.04727},
  year = {2024},
}

@inproceedings{sokal2005multivariate,
  title = {The multivariate Tutte polynomial (alias Potts model) for graphs and matroids.},
  author = {Sokal, Alan D and others},
  booktitle = {BCC},
  pages = {173--226},
  year = {2005},
}

@inproceedings{liu2014fptas,
  title = {FPTAS for counting monotone CNF},
  author = {Liu, Jingcheng and Lu, Pinyan},
  booktitle = {Proceedings of the Twenty-Sixth Annual ACM-SIAM Symposium on Discrete Algorithms},
  pages = {1531--1548},
  year = {2014},
  organization = {SIAM},
}

@article{lu2015fptas,
  title = {FPTAS for hardcore and Ising models on hypergraphs},
  author = {Lu, Pinyan and Yang, Kuan and Zhang, Chihao},
  journal = {arXiv preprint arXiv:1509.05494},
  year = {2015},
}

@article{bezakova2019approximation,
  title = {Approximation via correlation decay when strong spatial mixing fails},
  author = {Bez{\'a}kov{\'a}, Ivona and Galanis, Andreas and Goldberg, Leslie Ann and Guo, Heng and Stefankovic, Daniel},
  journal = {SIAM Journal on Computing},
  volume = {48},
  number = {2},
  pages = {279--349},
  year = {2019},
  publisher = {SIAM},
  timestamp = {Sun, 19 Jan 2025 00:00:00 +0100},
  biburl = {https://dblp.org/rec/journals/siamcomp/BezakovaGGGS19.bib},
  bibsource = {dblp computer science bibliography, https://dblp.org},
  doi = {10.1137/16M1083906},
  _bib2doi_selected = {dblp:/rec/journals/siamcomp/BezakovaGGGS19.bib},
  _bib2doi_confirmed = {true},
}

@article{trinks2016survey,
  title = {A survey on recurrence relations for the independence polynomial of hypergraphs},
  author = {Trinks, Martin},
  journal = {Graphs and Combinatorics},
  volume = {32},
  number = {5},
  pages = {2145--2158},
  year = {2016},
  publisher = {Springer},
  timestamp = {Thu, 04 Jun 2020 01:00:00 +0200},
  biburl = {https://dblp.org/rec/journals/gc/Trinks16.bib},
  bibsource = {dblp computer science bibliography, https://dblp.org},
  doi = {10.1007/s00373-016-1685-z},
  _bib2doi_selected = {dblp:/rec/journals/gc/Trinks16.bib},
  _bib2doi_confirmed = {true},
}

@article{chen2024spectral,
  title = {Spectral independence via stability and applications to Holant-type problems},
  author = {Chen, Zongchen and Liu, Kuikui and Vigoda, Eric},
  journal = {TheoretiCS},
  volume = {3},
  year = {2024},
  publisher = {Episciences. org},
  timestamp = {Mon, 09 Sep 2024 01:00:00 +0200},
  biburl = {https://dblp.org/rec/journals/theoretics/ChenLV24.bib},
  bibsource = {dblp computer science bibliography, https://dblp.org},
  doi = {10.46298/theoretics.24.16},
  _bib2doi_selected = {dblp:/rec/journals/theoretics/ChenLV24.bib},
  _bib2doi_confirmed = {true},
}

@article{feng2023swendsen,
  title = {Swendsen-Wang dynamics for the ferromagnetic Ising model with external fields},
  author = {Feng, Weiming and Guo, Heng and Wang, Jiaheng},
  journal = {Information and Computation},
  volume = {294},
  pages = {105066},
  year = {2023},
  publisher = {Elsevier},
  timestamp = {Mon, 03 Mar 2025 00:00:00 +0100},
  biburl = {https://dblp.org/rec/journals/iandc/FengGW23.bib},
  bibsource = {dblp computer science bibliography, https://dblp.org},
  doi = {10.1016/j.ic.2023.105066},
  _bib2doi_selected = {dblp:/rec/journals/iandc/FengGW23.bib},
  _bib2doi_confirmed = {true},
}

@article{gheissari2024spatial,
  title = {Spatial mixing and the random-cluster dynamics on lattices},
  author = {Gheissari, Reza and Sinclair, Alistair},
  journal = {Random Structures \& Algorithms},
  volume = {64},
  number = {2},
  pages = {490--534},
  year = {2024},
  publisher = {Wiley Online Library},
  timestamp = {Sat, 08 Jun 2024 01:00:00 +0200},
  biburl = {https://dblp.org/rec/journals/rsa/GheissariS24.bib},
  bibsource = {dblp computer science bibliography, https://dblp.org},
  doi = {10.1002/rsa.21191},
  _bib2doi_selected = {dblp:/rec/journals/rsa/GheissariS24.bib},
  _bib2doi_confirmed = {true},
}

@article{bandyopadhyay2008counting,
  title = {Counting without sampling: Asymptotics of the log-partition function for certain statistical physics models},
  author = {Bandyopadhyay, Antar and Gamarnik, David},
  journal = {Random Structures \& Algorithms},
  volume = {33},
  number = {4},
  pages = {452--479},
  year = {2008},
  publisher = {Wiley Online Library},
  timestamp = {Fri, 26 May 2017 01:00:00 +0200},
  biburl = {https://dblp.org/rec/journals/rsa/BandyopadhyayG08.bib},
  bibsource = {dblp computer science bibliography, https://dblp.org},
  doi = {10.1002/rsa.20236},
  _bib2doi_selected = {dblp:/rec/journals/rsa/BandyopadhyayG08.bib},
  _bib2doi_confirmed = {true},
}

@article{bencs2025barvinok,
  title = {Barvinok's interpolation method meets Weitz's correlation decay approach},
  author = {Bencs, Ferenc and Regts, Guus},
  journal = {arXiv preprint arXiv:2507.03135},
  year = {2025},
}

@book{von2003modern,
  title = {Modern computer algebra},
  author = {Von Zur Gathen, Joachim and Gerhard, J{\"u}rgen},
  year = {2003},
  publisher = {Cambridge university press},
}

@article{bencs2025optimal,
  title = {Optimal Zero-Free Regions for the Independence Polynomial of Bounded Degree Hypergraphs},
  author = {Bencs, Ferenc and Buys, Pjotr},
  journal = {Random Structures \& Algorithms},
  volume = {66},
  number = {4},
  pages = {e70018},
  year = {2025},
  publisher = {Wiley Online Library},
  timestamp = {Thu, 11 Sep 2025 01:00:00 +0200},
  biburl = {https://dblp.org/rec/journals/rsa/BencsB25.bib},
  bibsource = {dblp computer science bibliography, https://dblp.org},
  doi = {10.1002/rsa.70018},
  _bib2doi_selected = {dblp:/rec/journals/rsa/BencsB25.bib},
  _bib2doi_confirmed = {true},
}

\end{document}